\documentclass[12pt, a4paper, reqno]{amsart}
\usepackage[margin=1in]{geometry}
\usepackage{graphicx}
\usepackage{float}
\usepackage{amsmath}
\usepackage{amssymb}
\usepackage[colorlinks = true, linkcolor = blue, urlcolor  = blue, citecolor = red]{hyperref}
\usepackage[style=alphabetic,doi=false,isbn=false,url=false,eprint=false]{biblatex}
\usepackage[capitalize]{cleveref}
\usepackage{amsthm}
\newtheorem{theorem}{Theorem}[section]
\newtheorem*{theorem*}{Theorem}

\newtheorem{lemma}[theorem]{Lemma}
\newtheorem{corollary}{Corollary}[theorem]
\theoremstyle{definition}
\newtheorem{definition}{Definition}[section]
\theoremstyle{remark}
\newtheorem{remark}[theorem]{Remark}
\usepackage{float}
\usepackage{mathtools}
\usepackage{stackengine}
\usepackage{xcolor}

\title[Min-cuts partial order and quantum entanglement]{Canonical partial ordering from min-cuts and\\quantum entanglement in random tensor networks}

\author{Miao Hu}
\email{miao.hu@irsamc.ups-tlse.fr}
\address{Univ Toulouse, CNRS, LPT, Toulouse, France}

\author{Ion Nechita}
\email{nechita@irsamc.ups-tlse.fr}
\address{Univ Toulouse, CNRS, LPT, Toulouse, France}

\begin{document}
\newcommand{\seriescon}{\ensurestackMath{\stackinset{c}{0ex}{c}{0.2ex}{\scriptscriptstyle s}{\scriptstyle\cup}}}
\newcommand{\parallelcon}{\ensurestackMath{\stackinset{c}{0ex}{c}{0.2ex}{\scriptscriptstyle p}{\scriptstyle\cup}}}
\newcommand{\mincut}{\operatorname{mincut}}
\newcommand{\maxflow}{\operatorname{maxflow}}

\begin{abstract}
    The \emph{max-flow min-cut theorem} has been recently used in the theory of random tensor networks in quantum information theory, where it is helpful for computing the behavior of important physical quantities, such as the entanglement entropy. In this paper, we extend the max-flow min-cut theorem to a relation among different \emph{partial orders} on the set of vertices of a network and introduce a new partial order for the vertices based on the \emph{min-cut structure} of the network. We apply the extended max-flow min-cut theorem to random tensor networks and find that the \emph{finite correction} to the entanglement R\'enyi entropy arising from the degeneracy of the min-cuts is given by the number of \emph{order morphisms} from the min-cut partial order to the partial order induced by non-crossing partitions on the symmetric group. Moreover, we show that the number of order morphisms corresponds to moments of a graph-dependent measure which generalizes the free Bessel law in some special cases in free probability theory.
\end{abstract}

\maketitle

\tableofcontents

\section{Introduction}

The study of cuts and flows in graphs is a cornerstone of both combinatorial optimization and mathematical physics, with deep implications for network theory, computer science, and quantum information. In a graph, a set of edges is called a \emph{cut} if by removing these edges, the number of connected components of the graph increases. Cuts are often mentioned in the study of graphs themselves, where several families of cuts of a graph with interesting properties are studied. For example, the family of all cuts defines a vector space called the cut space of the graph, and the minimum cuts (min-cut), which are members of the family of the cuts with a minimum number of edges, are found to be in duality with the maximum amount of flow (max-flow) one can send from a source to a sink through the graph \cite{diestel_graph_2017, ford_flows_1974}. 

More recently, cuts have also appeared in the field of quantum information and quantum many body physics, where graphs become tools for representing states in a complicated Hilbert space by encoding the locality structure of the interactions. One of the central concepts in many-body physics is the \emph{area law} for entanglement entropy, which states that the entanglement entropy of a subsystem typically scales with the size of its boundary rather than its volume \cite{eisert_area_2010}. This phenomenon is especially prominent in ground states of gapped local Hamiltonians, where correlations decay rapidly with distance. In the context of tensor network states, such as matrix product states (MPS) and projected entangled pair states (PEPS) \cite{orus2019tensor}, the area law manifests naturally: the entanglement entropy between two regions is bounded above by the logarithm of the bond dimension times the number of edges (or bonds) crossing the boundary between the regions. In graphical terms, this boundary corresponds to a \emph{cut} in the underlying network, and the minimal number of edges that must be severed to separate the two regions --- the \emph{minimal cut} --- directly determines the maximal entanglement entropy that can be supported by the tensor network. Thus, the area law in tensor network states is intimately connected to the combinatorial structure of minimal cuts in the associated graph, providing a bridge between physical entanglement properties and graph-theoretic quantities.

Random tensor networks are probability distributions on the set of tensor networks, usually defined by assigning a random tensor to each vertex of the network. 
In \cite{hayden_holographic_2016}, the authors have investigated such models and discovered a new relation between the entanglement in a subregion on the boundary of the network and the min-cuts of the network homologous to the subregion. This relation resembles the \emph{Ryu-Takayanagi} formula, conjectured for a conformal field theory dual to a quantum gravity theory in the anti-de Sitter space (AdS/CFT correspondence), which is increasingly important for both high-energy physics and qunatum information as it opens a path to study geometry from entanglement as well as the converse \cite{ryu_holographic_2006, rangamani_holographic_2017, chen_quantum_2022}. Random tensor networks then become one of most tractable models with a holographic property like AdS/CFT correspondence, which attracts attention for studying their entanglement properties \cite{cheng_random_2022, fitter_max-flow_2024, penington_fun_2023, dong_holographic_2021, yang_entanglement_2022, kudler-flam_negativity_2022, kudler-flam_renyi_2024, nezami_multipartite_2020, akers_reflected_2022, akers_entanglement_2023}. 

This paper is also motivated by the study of random tensor networks, with a particular emphasis on their relationship with the family of min-cuts. We introduce a partial ordering of vertices derived from the relations among the min-cuts in the cut space of a graph. We show that min-cuts contain information more than the dominant contribution to the entanglement entropy of a subregion on the boundary of a random tensor network, which is known to be proportional to the number of edges in a min-cut homologous to the subregion \cite{hayden_holographic_2016}. We will show that a finite correction to the entanglement entropy can also be computed using a certain relation among the homologous min-cuts. It turns out that this relation leads to a partial ordering of vertices which is equivalent to the partial order relations considered in \cite{fitter_max-flow_2024}. This equivalence then gives a proof to the independence of the partial order relations in \cite{fitter_max-flow_2024} from their choice of max-flows, which is necessary for justifying their method of computing the entanglement entropy in random tensor networks but is not provided with a proof by them.

There are a few notations that are used throughout this paper. For a graph $G$, the set of vertices is denoted by $V(G)$ and the set of edges is denoted by $E(G)$. The vertices of a graph are denoted by the letters $x, y, z$, the directed edges by $(x, y)$, and the undirected edges by $e$ or $e_{xy}$ if we know the vertices at the two ends of the edge are $x$ and $y$. For each undirected graph $G$, we associate it with a bi-directed graph $\overset{\rightleftharpoons}{G}$ by mapping each undirected edge $e_{xy}\in E(G)$ to two directed edges $(x, y), (y, x)\in E(\overset{\rightleftharpoons}{G})$, as given from \cref{eqn:undirected and bidirected}. For a subset $S\subseteq V(G)$, we denote its complement set by $\Bar{S}$, where $\Bar{S}$ is also equal to $V(G)-S$. Other notation used throughout the paper will be introduced in subsequent sections.

We summarize our main theorems below. Consider a flow network $(G, c, s, t)$ where $G$ is a graph, $c:E(G)\to \mathbb{R}^+$ assigns each edge $e$ a capacity $c_e$, and $s, t\in V(G)$ are two distinguished vertices that we call the source and the sink respectively. A flow $f:E(\overset{\rightleftharpoons}{G})\to \mathbb{R}$ (see \cref{def:flows}) induces a partial order $\leq_f$ (see \cref{def:transitive closure}) on the set of vertices $V(G)$ such that $x\leq_f y$ if and only if there exists a directed path $P$ from $x$ to $y$ such that $f(x_i, x_{i+1})>-c_{e_{x_ix_{i+1}}}$ for each $(x_i, x_{i+1})\in P$. Our main result refines the classical max-flow min-cut theorem by showing that the partial order induced by any max-flow coincides with a canonical partial order derived from the min-cuts, and is independent of the specific choice of max-flow. Then, we have obtained the following theorem, which refines the max-flow min-cut theorem. 
\begin{theorem*}[Informal]
    Given a flow network $(G, c, s, t)$, the partial order $\leq_f$ induced from a max-flow $f$ is equal to the canonical partial order $\leq_{\mincut}$ induced from the min-cuts of the flow network (see \cref{def:canonical parital ordering}), and hence is independent from the choice of the max-flow $f$, i.e. 
    \begin{equation*}
        \forall f, f^{\prime} \text{ max-flows,}\quad x\leq_fy \iff x\leq_{\mincut} y \iff x\leq_{f^{\prime}} y.
    \end{equation*}
\end{theorem*}
Our second main theorem concerns the application of the above theorem to the computation of the entanglement R\'enyi entropy of random tensor networks. More precisely, consider a flow network $(G, c, s, t)$, where we associate each edge $e$ with a $D_e\times D_e$ dimensional Hilbert space $\mathcal{H}_e\coloneqq\mathbb{C}^{D_e}\otimes\mathbb{C}^{D_e}$ (or equivalently we associate each vertex $x$ with a $\prod_{y:\, e_{xy}\in E(G)} D_{e_{xy}}$ dimensional Hilbert space $\mathcal{H}_x\coloneqq \bigotimes_{y:\, e_{xy}\in E(G)}\mathbb{C}^{D_{e_{xy}}}$) such that $D_e$ is parametrized by $D_e=D^{c_e}$. A random tensor network on this flow network is a (unnormalized) random state defined as 
\begin{equation*}
    |\psi_G\rangle \coloneqq \left\langle\bigotimes_{e\in E(G-\{s, t\})}\Omega_e\right|\left.\bigotimes_{x\in V(G-\{s, t\})}g_x\right\rangle \in \mathcal{H}_s \otimes \mathcal{H}_t
\end{equation*}
where $G-\{s, t\}$ is the subgraph of $G$ induced by cutting the source vertex and the sink vertex, $|g_x\rangle\in \mathcal{H}_x$ is a random Gaussian state at each vertex $x$, $|\Omega_e\rangle = \frac{1}{\sqrt{D_e}}\sum_{i=1}^{D_e}|ii\rangle\in\mathcal{H}_e$ is an EPR pair at each edge $e$. Then the entanglement R\'enyi entropy of the above random state with respect to the bipartition $\mathcal{H}_s$ and $\mathcal{H}_t$ is given in the following theorem.
\begin{theorem*}[Informal]
    Let $(G, c, s, t)$ be a flow network on which each edge $e$ is associated with a $D^{c_e}\times D^{c_e}$ dimensional Hilbert space and there is a corresponding random state $|\psi_G\rangle\in\mathcal{H}_s\otimes\mathcal{H}_t$. Then when $D$ is large, the entanglement R\'enyi entropy of order $n$ of the state $|\psi_G\rangle$ has a typical value given as 
    \begin{equation*}
        \mathsf{H}_n \approx \mincut(G)\log D - \frac{\log |\mathrm{Ord}(V(G), \mathcal{S}_n)|}{n-1} + \dots,
    \end{equation*}
    where $\mathrm{Ord}(V(G), \mathcal{S}_n)$ is the set of order morphisms from the poset $(V(G), \leq_{\mincut})$ to the poset $(\mathcal{S}_n, \leq_{\mathrm{nc}})$, where $\leq_{\mathrm{nc}}$ is the partial order on the symmetry group $\mathcal{S}_n$ introduced in free probability theory \cite[Notation 23.21]{nica_lectures_2006} (see the review in \cref{appendix:free-probability}), with fixed boundary conditions on $s$ and $t$, i.e.   
    \begin{multline*}
        \mathrm{Ord}(V(G), \mathcal{S}_n):=\{\boldsymbol{\beta}:V(G)\to\mathcal{S}_n|\\
        \boldsymbol{\beta}(s)=\mathrm{id},\  \boldsymbol{\beta}(t)=(1\cdots n), \text{ and}\\
        \boldsymbol{\beta}(x)\leq_{\mathrm{nc}}\boldsymbol{\beta}(y)\text{ iff }x\leq_{\operatorname{mincut}}y\}, 
    \end{multline*}
    and ``$\dots$'' represents the other sub-leading terms. Note that the above expression for the entanglement R\'enyi entropy becomes exact when $D\to\infty$. 
\end{theorem*}
Our final main theorem relates the size of the set of order morphisms with the moments of a graph-dependent measure, which is also the empirical eigenvalue distribution of a random tensor network. 
\begin{theorem*}[Informal]
    If the Hasse diagram which represents the poset $(V(G), \leq_{\mincut})$ is series-parallel (see \cref{sec:series-parallel}), then $|\mathrm{Ord}(V(G), \mathcal{S}_n)|$ is the $n$-th moment of the graph-dependent measure, from \cref{def:measure}, for the Hasse diagram. 
\end{theorem*}
The more formal and general theorems are presented in the main content of this paper, where we consider a more general parametrization for the dimension of the Hilbert space on each edge, that is, we consider $D_e=r_e D^{c_e}$ for each edge $e$ in a flow network and $r_e\in\mathbb{R}^+$ is not necessarily equal to 1. 

\bigskip

This paper is organized as follows. In \cref{sec:partial-orders}, we introduce the notion of partial orders induced by flows and min-cuts in flow networks, and establish the equivalence between the partial order from any max-flow and the canonical partial order from min-cuts. In \cref{sec:random-tensor-networks}, we apply these results to random tensor networks, showing how the structure of min-cuts determines the leading and subleading contributions to the entanglement R\'enyi entropy, and we provide a combinatorial formula for the entropy in terms of order morphisms. \cref{sec:graph-dependent measures} discusses the special case of series-parallel posets, where the set of order morphisms admits a moment interpretation via a graph-dependent measure. Technical details and background on free probability theory are collected in \cref{appendix:free-probability}.

\section{Partial orders for vertices of a network}
\label{sec:partial-orders}
Network flow problems consist of the study of transporting commodities through a network. Naturally, as one can speak of whether one vertex precedes another in a network based on whether it is possible to have a flow from one vertex to another, a flow in the network induces a partial order relation among the vertices. This section is devoted to the mutual relations among these partial orders. We start this section by formally defining flow networks and relevant basic notations. A flow network $(G, c, s, t)$ is a directed graph $G$ with two distinguished vertices (the source $s$ and the sink $t$) and a capacity $c:E(G)\to\mathbb{R}_+$ associated with each of its edges. One can imagine these edges as tunnels through which some commodities can flow along their directions. We can extend this definition of flow networks and consider flow networks of undirected graphs where commodities are allowed to flow in both directions of the edges. A flow on a undirected graph $G=(V(G), E(G))$ with capacities $c_e$ associated with each of its edges $e$ corresponds to a flow on the corresponding \emph{bi-directed} graph $\overset{\rightleftharpoons}{G}=(V(\overset{\rightleftharpoons}{G}), E(\overset{\rightleftharpoons}{G}))$, where  
\begin{equation}
    \label{eqn:undirected and bidirected}
    V(G)=V(\overset{\rightleftharpoons}{G}) \quad\text{and}\quad e_{xy}\in E(G) \iff (x, y),\, (y, x)\in E(\overset{\rightleftharpoons}{G}),
\end{equation}
and the capacities of the edges in the bi-directed graph are symmetrized and taken to be the same as those in the undirected graph, i.e.
\begin{equation*}
    c_{e_{xy}} = c_{(x, y)} = c_{(y, x)}.
\end{equation*}
We can now proceed to give a formal definition for a flow as well as the optimization of flows --- the max-flow. 

\begin{figure}[t]
    \centering
    \includegraphics[scale=1.15]{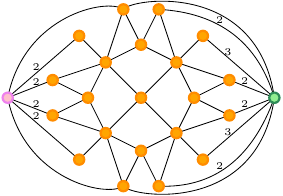}
    \caption{A flow network with the pink vertex as the source and the green vertex as the sink. The capacity for each edge is assumed to be 1 except for a few edges where we annotate their capacities by non-negative numbers.}
    \label{fig:example of a flow network}
\end{figure}

\subsection{Maximum flows}
Given a graph $G=(V(G), E(G))$ with two distinguished vertices $s$ and $t$ in $V(G)$, which we call the source and the sink respectively, and a set of edge capacities $\{c_e\}_{e\in E(G)}$, a flow is defined as follows:  
\begin{definition}
    \label{def:flows}
    A \emph{flow} from the source $s$ to the sink $t$ is a map $f:E(\overset{\rightleftharpoons}{G})\to\mathbb{R}$ such that 
    \begin{itemize}
        \item it is skew symmetric, i.e. $f(x, y)=-f(y, x)$;\footnote{A flow $f^{\prime}$ for directed graphs is often defined by replacing the skew symmetric property with the non-negativity (i.e. $f^{\prime}(x, y)\geq0$ for each directed edge $(x, y)$). Then, such a flow is related to the flow $f$ for undirected (bidirected) graphs in our definition by $f(x, y)=f^{\prime}(x, y)-f^{\prime}(y, x)$.}
        \item the flow through each edge is bounded by the capacity of the edge, i.e. $f(x, y)\leq c_{(x, y)}$;
        \item the total flow entering and exiting a vertex is conserved except at the source and sink, i.e. $\sum_{y:(x, y)\in E(\overset{\rightleftharpoons}{G})} f(x, y) = 0$ for $x\neq s, t$.
    \end{itemize}
    The value of a flow $|f|_c$ is the net flow out of the source or the net flow into the sink, that is, 
    \begin{equation*}
        |f|_c \coloneqq \sum_{y:(s, y)\in E(\overset{\rightleftharpoons}{G})} f(s, y) = \sum_{y:(y, t)\in E(\overset{\rightleftharpoons}{G})} f(y, t).
    \end{equation*}
    A flow is called a \emph{maximum flow (or max-flow)} if its value is maximal.
\end{definition}

\begin{figure}[htb]
    \centering
    \includegraphics[scale=1.15]{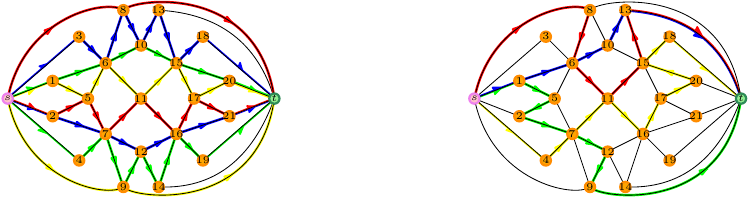}
    \caption{A max-flow (left) and a flow (right) on a network. The capacity for each edge is the same as in \cref{fig:example of a flow network} and each vertex here is labeled by a number. A highlighted path represents a transportation of one unit of flow from the source to the sink. The map $f$ evaluated at a particular edge is the number of highlighted paths going forward minus the number of highlighted paths going backward through that edge. Note that for the left, $|f|_c=8$ and for the right, $|f|_c=4$.}
    \label{fig:flow}
\end{figure}

Network flows are natural models of transportation through routes, channels, or other networks of the like in mathematics, and the max-flow is the optimal solution. In theoretical computer science, the problem of finding a max-flow is often formulated as a linear program and there is a large body of literature focusing on exploring efficient algorithms \cite{korte_combinatorial_2018} that solve this linear program. In graph theory, the network flow problem is also discussed extensively as a certain type of flows (such as circulations, i.e. a flow with $|f|_c=0$, and group-valued flows) reveals the connectivity of graphs \cite{diestel_graph_2017}. Here, we are interested in understanding flows from the partial orders for the vertices that can be extracted from flows. Naturally, there exists a partial order for a flow based on whether or not there are capacities left for a path between two vertices given that some capacities of the path have already been used by the flow, and it leads to the following definition. 
\begin{definition}
    \label{def:transitive closure}
    Given a flow $f$ in a network, we define the binary relation $R_f$ on the set of vertices as 
    \begin{equation*}
        x\,R_f\,y \quad \text{if and only if} \quad f(x, y)>-c_{(x, y)}.
    \end{equation*}
    We also introduce the partial order $\leq_f$ of vertices as the transitive closure of the relation $R_f$:
    \begin{equation*}
        x\leq_f y \quad \text{if and only if} \quad \exists z_1, z_2, \dots\,: x\,R_f\,z_1\,R_f\,z_2\cdots\,R_f\,y.
    \end{equation*}
\end{definition}
\begin{remark}
    \label{remark: flow partial order}
    Note that the binary relation $R_f$ is not symmetric, i.e. $x R_f y$ does not necessarily imply that $y R_f x$, and the relation here is the reverse of the relation in \cite{cottle_structure_1980}. The relation $\leq_f$ is transitive by definition and hence a partial order on the set of all equivalence classes of $V(G)$ with respect to the equivalence relation $=_{f}$, where $x=_fy$ if and only if $x\leq_fy$ and $y\leq_fx$ (for simplicity, we simply say $\leq_f$ is a partial order on $V(G)$ afterwards). It is worth noting that if $f$ is a max-flow, then $s\leq_f x\leq_f t$ for any $x\in V(G)$, since for any path from $s$ to $t$ via $x$, $f$ must be at least non-negative everywhere on the path for it to be a max-flow. 
\end{remark}
As an example, the partial orders for the flows in \cref{fig:flow} are respectively: 
\begin{itemize}
    \item left: 
    \begin{equation}
    \label{eqn: partial order from a max-flow}
        \left\{\begin{array}{l}
            s=_f3=_f4 \\
            s\leq_f1\leq_f5\leq_f6\leq_f8\leq_f10\leq_ft \\
            \qquad\qquad\qquad\quad6\leq_f11\leq_ft \\ 
            s\leq_f2\leq_f5\leq_f7\leq_f11\leq_ft \\ 
            \qquad\qquad\qquad\quad7\leq_f9\leq_f12\leq_ft \\ 
            t=_f13=_f\dots=_f21
        \end{array}\right.,
    \end{equation}
    \item right: 
    \begin{equation*}
        s=_f1=_f\dots=_f21=_ft
    \end{equation*}
\end{itemize}
We will return to the partial order defined here in \cref{sec:max-flow min-cut theorem}.

\subsection{Minimum cuts}
In this subsection, we discuss the dual picture of flows in linear programming --- the cuts. We say a set of edges is a cut of a graph if by removing these edges from the graph, it increases the number of connected components of the graph. Equivalently, any cut can be regarded as a bipartition of the set of vertices.
\begin{definition}
    \label{def:cuts}
    Given a flow network $(G, c, s, t)$, \emph{an $s-t$ cut} is a bipartition $(S, T)$ of $V(G)$ such that $s\in S$, $t\in T$. The set of edges corresponding to the bipartition which disconnects $s$ from $t$ is given as  
    \begin{equation}
        \label{eqn:cut set}
        (S, T) \coloneqq \{(x, y)\in E(\overset{\rightleftharpoons}{G})|x\in S\text{ and }y\in T\}. 
    \end{equation}
    The capacity of an $s-t$ cut $|(S, T)|_c$ is defined as the sum of the capacities of the edges included in the cut set, that is, 
    \begin{equation*}
        |(S, T)|_c\coloneqq\sum_{(x, y)\in (S, T)}c_{(x, y)}\in\mathbb{R}_{+}.
    \end{equation*}
    An $s-t$ cut is called a \emph{minimum $s-t$ cut (or min-cut)} if its capacity is minimal.
\end{definition}

\begin{figure}[htb]
    \centering
    \includegraphics[scale=1.15]{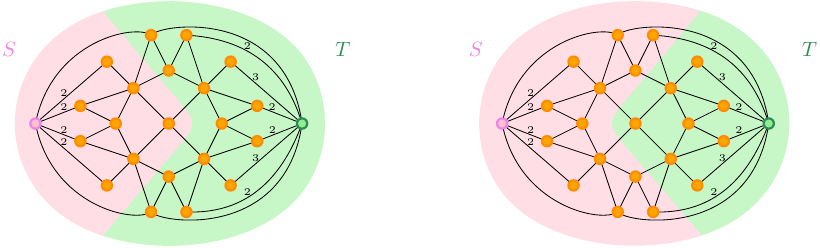}
    \caption{A min-cut (left) and a cut (right) for a network. Note that for the left, $|(S, T)|_c=8$ and for the right, $|(S, T)|_c=12$.}
    \label{fig:cut}
\end{figure}

Indeed, finding a min-cut for a flow network is the dual optimization problem of finding a max-flow for the same network. This duality has been firstly discovered by Ford and Fulkerson \cite{ford_flows_1974}, under the name of max-flow min-cut theory, which will be introduced in \cref{sec:max-flow min-cut theorem}. Here, similar to the discussion about flows, we are interested in understanding cuts from the partial orders for the vertices that can be extracted from mutual relations among the cuts. In particular, we are motivated by discovering the partial orders for the vertices implied by the duality between max-flows and min-cuts.  
\begin{definition}
    \label{def:canonical parital ordering}
    Given a flow network $(G, c, s, t)$ and a vertex $x\in V(G)$, the \emph{finest min-cut} that does not separate $x$ from the source is a bipartition $(S_x, \Bar{S}_x)$ with $S_x$ defined as 
    \begin{equation}
        \label{eqn:minimal min-cut containing x}
        S_x \coloneqq 
            \bigcap_{(S, T): \text{min-cut and $x\in S$}} S, 
    \end{equation}
    where interpret the empty intersection to be $V(G)$. The canonical partial order $\leq_{\mincut}$ on the set of vertices $V(G)$ is then defined as 
    \begin{equation*}
        x\leq_{\mincut} y \quad \text{if}\quad S_x \subseteq S_y.
    \end{equation*}
\end{definition}
\begin{remark}
    Note that $x\in S_x$ for all $x\in V(G)$ and it is trivial to show that $\leq_{\mincut}$ is indeed a partial order on the set of all equivalence classes of $V(G)$ with respect to the equivalence relation $=_{\mincut}$, where $x=_{\min}y$ if $x\leq_{\mincut}y$ and $y\leq_{\mincut}x$, since it directly follows from the properties of the inclusion $\subseteq$ (for simplicity, we simply say $\leq_{\mincut}$ is a partial order on $V(G)$ afterwards). Moreover, this partial order $\leq_{\mincut}$ reveals an interesting structure on the elements in $S_x$, which is summarized in the following lemma. 
\end{remark}

\begin{figure}[htb]
    \centering
    \includegraphics[scale=1.15]{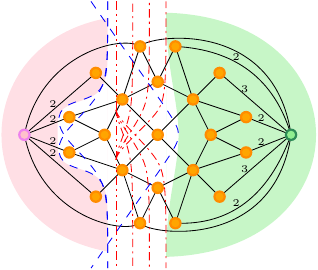}
    \caption{The degeneracy of min-cuts. The pink area corresponds to the minimal set of vertices $S$ such that $(S, T)$ is a min-cut and the green area corresponds to the minimal set of vertices $T$ (equivalently, the maximal set of vertices $S$) such that $(S, T)$ is a min-cut.  The other min-cuts which are related to the set $S_x$ in \cref{def:canonical parital ordering} are annotated by dashed curves.}
    \label{fig:set of min-cuts}
\end{figure}

\begin{lemma} 
\label{lemma:minimal min-cuts containing x}
The finest min-cut that does not separate $y$ from the source corresponds to the set of vertices that precede $y$, i.e. 
    \begin{equation*}
        S_y = \{x\in V(G)|x\leq_{\mincut} y\}.
    \end{equation*}
\end{lemma}
\begin{proof}
    We shall use the property that 
    \begin{center}
        if $(S_1, T_1)$ and $(S_2, T_2)$ are both min-cuts, then $(S_1\cap S_2, T_1\cup T_2)$ is also a min-cut, 
    \end{center}
    from \cite[Corollary 3]{cottle_structure_1980} or \cite[Corollary 5.3]{ford_flows_1974}. Then, with this property, we can find that for any $y\in V(G)$, either $(S_y, \Bar{S}_y)$ is a min-cut or $S_y= V(G)$. Suppose that now $x\in S_y$. Then $S_x\subseteq S_y$ since either $(S_y, \Bar{S}_y)$ satisfies the condition for $(S, T)$ required by \cref{eqn:minimal min-cut containing x} or $S_y=V(G)$ and hence $x\leq_{\mincut} y$ by definition. Conversely, suppose that $x\leq_{\mincut} y$. Then $S_x\subseteq S_y$ by definition and hence $x\in S_y$ since $x\in S_x\subseteq S_y$. Thus, we have showed that 
    \begin{equation*}
        x\in S_y \iff x\leq_{\mincut} y, 
    \end{equation*}
    which is equivalent to the statement in the lemma. 
\end{proof}

\subsection{The max-flow min-cut theorem}
\label{sec:max-flow min-cut theorem}
Here, we introduce one of the most important results in combinatorial optimization --- the max-flow min-cut theorem. This theorem is essential in proving our result concerning the max-flow min-cut duality for the partial orders for the vertices. 
\begin{theorem}[The max-flow min-cut theorem \cite{ford_flows_1974}]
    \label{theorem:max-flow min-cut theorem}
    For a flow network $(G, c, s, t)$, the value of a maximum flow from the source to the sink is equal to the capacity of a minimum $s-t$ cut. That is, 
    \begin{equation*}
        \maxflow(G)\coloneqq\max_{f}|f|_c = \min_{(S, T)}|(S, T)|_c\eqqcolon\mincut(G).
    \end{equation*}
\end{theorem}
\begin{remark}
    \label{remark: max-flow min-cut theorem}
    By the max-flow min-cut theorem, it is also easy to find that 
    \begin{equation*}
        \forall f,\, \forall (S, T), \quad |f|_c\leq\maxflow(G)=\mincut(G)\leq|(S, T)|_c. 
    \end{equation*}
    It then implies that $|f|_c=|(S, T)|_c$ if and only if $f$ is a max-flow and $(S, T)$ is a min-cut. Essentially, finding 
    \begin{equation*}
        \begin{split}
            \max|f|_c \quad &\text{subject to} \\
            &f:\text{ a flow},
        \end{split}
    \end{equation*}
    is a linear program. Its dual program is to find \cite{naves_notes_nodate}
    \begin{equation*}
        \begin{split}
            \min\textstyle{\sum_{e_{xy}\in E(G)}}c_{e_{xy}}d(x, y) \quad&\text{subject to} \\
            &d:\text{ a metric on $V(G)$}, \\
            &d(s, t)\geq1.
        \end{split}
    \end{equation*}
    This dual program has an optimum solution when $d$ is such that $d(x, y)=1$ for each pair of $x$ and $y$ separated by a $s-t$ cut set and $d(x, y)=0$ otherwise, and hence min-cuts appear naturally in the network flow optimization problem. One can learn more about the duality for linear programs in \cite{korte_combinatorial_2018} and its application to the network flow problem in \cite{naves_notes_nodate}. 
\end{remark}

While \cref{theorem:max-flow min-cut theorem} states that the numerical values of a max-flow and a min-cut are equal, we prove next that the partial order on the set of vertices induced by an arbitrary flow is identical to the one induced by the min-cuts. Indeed, the max-flow min-cut theorem implies the equivalence among these partial orders, and the following result can be understood as a refinement of the max-flow min-cut theorem. 
\begin{theorem}
    \label{theorem:equivalence of all of the partial orderings}
    Given a flow network $(G, c, s, t)$ and an arbitrary max-flow $f$, the partial order $\leq_f$ is equal to the canonical partial order $\leq_{\mincut}$. That is, 
    \begin{equation*}
        \forall f \text{ max-flow},\quad x\leq_f y \iff x\leq_{\mincut} y.
    \end{equation*}
\end{theorem}
\begin{proof}
    At first, we need to introduce the notion of the closure for a binary relation on the set of vertices $V(G)$ of a graph $G$. We say a subset of vertices $C$ is a \emph{closure} for a binary relation $R$ if and only if $y\in C$ and $x\,R\,y$ imply that $x\in C$. Note that, there exist multiple closures for one binary relation and $V(G)$ is always a closure. Then, according to \cite[Theorem 1 and Proposition 4]{cottle_structure_1980}, a closure $C$ for the partial order $\leq_f$ which contains the source but the sink is a set of vertices equivalent to a minimum $s-t$ cut. That is,
    \begin{equation*}
        \text{$C$ is a closure: $s\in C$ and $t\notin C$} \iff \text{$(C, \Bar{{C}})$ is a minimum $s-t$ cut,}
    \end{equation*}
    where the definition for $(C, \Bar{C})$ resembles that in \cref{eqn:cut set}. This is because, if $C$ is such a closure and we pick an edge $(x, y)\in(C, \Bar{C})$, then $y \nleq_f x$. Otherwise, $y\leq_f x$ would imply that $y\in C$ since $C$ is a closure and $x\in C$. It then contradicts with the fact that $y\in \Bar{C}$. Since $y \nleq_f x$ for each $(x, y)\in(C, \Bar{C})$, we can find that $f(y, x)=-c_{(y, x)}$ by \cref{def:transitive closure}. Then, by the conservation law and skew symmetry for flows, the value of the flow is given as  
    \begin{equation*}
        |f|_c=\sum_{(x, y)\in (C, \Bar{C})}f(x, y)=|(C, \Bar{C})|_{c}
    \end{equation*}
    By \cref{remark: max-flow min-cut theorem}, the above is true if and only if $(C, \Bar{C})$ is a min-cut. Conversely, suppose $(C, \Bar{C})$ is a min-cut. Then, by \cref{remark: max-flow min-cut theorem}, for all $(x, y)\in (C, \Bar{C})$, we have $f(x, y)=c_{(x, y)}$ and hence $y\nleq_f x$. As any path starting from a vertex in $C$ ending at a vertex not in $C$ must intersect with $(C, \Bar{C})$, this implies that for any $x^{\prime}\in C$ and any $y^{\prime}\notin C$, it is impossible to find $y^{\prime}\leq_f x^{\prime}$. In other words, $C$ is a set containing all $x$ such that $x\leq_f x^{\prime}$, for any $x^{\prime}\in C$. Hence it is a closure. In addition, it contains $s$ but $t$ since the min-cut $(C, \Bar{C})$ separates $s$ from $t$. 

    We then prove that $S_y = \{x\in V(G)|x\leq_f y\}$. Suppose that $x_1\in\{x\in V(G)|x\leq_f y\}$ and there is a vertex $x_2$ such that $x_2\leq_f x_1$. Then $x_2\leq_f x_1\leq_f y$. By transitivity, $x_2\leq_f y$ and hence $x_2\in\{x\in V(G)|x\leq_f y\}$. This shows that the set $\{x\in V(G)|x\leq_f y\}$ is a closure. This set also contains $s$ since $s\leq_f y$ always holds (recall \cref{remark: flow partial order}). Hence depending on whether $t$ is in or not in this set, it is either equivalent to a minimum $s-t$ cut, which we have just proved, or equal to $V(G)$ itself, since $V(G)=\{x\in V(G)|x\leq_f t\}\subseteq\{x\in V(G)|x\leq_f y\}\subseteq V(G)$ when $t\leq_f y$. In the case where $\{x\in V(G)|x\leq_f y\}\neq V(G)$, since $S_y$ is defined as the intersection of the sets $S$ such that each of them is equivalent a minimum $s-t$ cut and also contains $y$, the closure $\{x\in V(G)|x\leq_f y\}$ is of course one candidate and we must have  
    \begin{equation*}
        S_y\subseteq \{x\in V(G)|x\leq_f y\}.
    \end{equation*}
    Moreover, for any closure (min-cut) $C$, including $S_y$, which contains $y$, it must be true that $x_1\leq_f y$ implies that $x_1\in C$ by the definition of the closure. Hence, $x_1\leq_f y$ implies that $x_1\in S_y$ and 
    \begin{equation*}
        \{x\in V(G)|x\leq_f y\}\subseteq S_y.
    \end{equation*}
    We have proven that $S_y = \{x\in V(G)|x\leq_f y\}$. In the case where $\{x\in V(G)|x\leq_f y\}=V(G)$, we have $s\leq_f y$ and $t\leq_f y$. This implies that if $C$ is a closure containing $y$, then it must also contain $s$ and $t$. Such a closure cannot be equivalent to a minimum $s-t$ cut. Thus, there does not exist a min-cut $(S, T)$ such that $y\in S$. Then, $S_y=V(G)$ by \cref{def:canonical parital ordering} and hence $S_y=\{x\in V(G)|x\leq_f y\}$. 
    
    Finally, we prove the theorem. Suppose $x\leq_f y$. Then $x\in S_y=\{x\in V(G)|x\leq_f y\}$ and hence $x\leq_{\mincut} y$ since $S_y = \{x\in V(G)|x\leq_{\mincut} y\}$ by \cref{lemma:minimal min-cuts containing x}. The converse can be shown by the same logic. 
\end{proof}

\begin{corollary}
    The max-flow induced partial order $\leq_f$ is independent from the choice of the max-flow, i.e. 
    \begin{equation*}
        \forall f, f^{\prime} \text{ max-flows}, \quad x\leq_f y \iff x\leq_{f^{\prime}} y. 
    \end{equation*}
\end{corollary}
We will present an application of \cref{theorem:equivalence of all of the partial orderings} in \cref{sec:order morphisms}, where we are interested in computing a finite correction term to the entanglement R\'enyi entropy of random tensor networks. 

\subsection{Quotient graphs and Hasse diagrams}

\begin{figure}[b]
    \centering
    \includegraphics[scale=1.15]{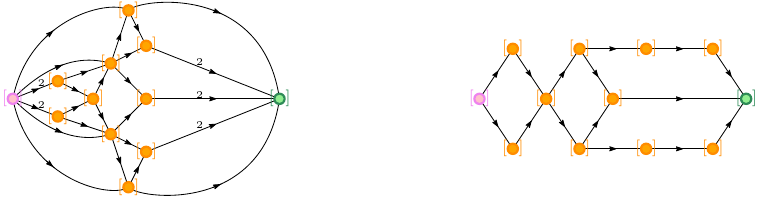}
    \caption{The quotient graph (left) and the Hasse diagram (right) for the flow network in \cref{fig:example of a flow network}.}
    \label{fig:Hasse diagram}
\end{figure}

In this subsection, we introduce two graphical representations for the partial order $\leq_{\mincut}$, namely the quotient graph and the Hasse diagram. The former is a \emph{directed acyclic graph} obtained from taking the quotient of a graph $G$ with respect to the equivalence relation $=_{\mincut}$ on the set of vertices $V(G)$ and the latter is a subgraph of the former, which is commonly used to represent partially ordered sets (posets). More precisely, for each edge $e_{xy}$ in $G$, we can replace it with a directed edge $(x, y)$ if $x\leq_{\mincut}y$. It then gives us a \emph{directed graph} possibly with some oriented cycles (also known as strongly connected components). Each oriented cycle of this directed graph indeed corresponds to a block of vertices at the same order. If we group vertices of equal orders as block vertices, we can then obtain a directed acyclic graph, which is the quotient graph of the directed graph with respect to the equivalence relation $=_{\mincut}$. Hence, for each flow network $(G, c, s, t)$, we can associate it with a directed acyclic graph $\vec{G}$, which we simply call the \emph{quotient graph}, defined as follows: 
\begin{equation}
    \label{eqn:quotient graph}
    \begin{split}
        V(\vec{G})&=\{[x]\,\boldsymbol{|}\,[x]\coloneqq \{y\in V(G)|x=_{\mincut}y\}\}\quad\text{and}\\
        E(\vec{G})&=\{([x], [y])\,\boldsymbol{|}\,\exists x\in[x], y\in[y]: e_{xy}\in E(G)\}.
    \end{split}
\end{equation}
The Hasse diagram that represents the poset $(V(G), \leq_{\mincut})$ is a subgraph of the quotient graph $\vec{G}$ such that there is no edge connecting $[x]$ and $[z]$ if there exists $[y]$ such that $[x]\leq_{\mincut} [y]\leq_{\mincut}[z]$. As an example, the quotient graph and the Hasse diagram associated to the flow network in \cref{fig:example of a flow network} are given in \cref{fig:Hasse diagram}, and it is not hard to see that the Hasse diagram in \cref{fig:Hasse diagram} is the same as the Hasse diagram that represents the poset $(V(G), \leq_f)$ when $f$ is a max-flow (see \cref{eqn: partial order from a max-flow}), which is consistent with \cref{theorem:equivalence of all of the partial orderings}.

A natural question one may ask is that what Hasse diagrams can be obtained from flow networks. Equivalently, we need to answer the question that given an arbitrary Hasse diagram, whether there exists a flow network such that its canonical partial order for the vertices $\leq_{\mincut}$ is represented by the given Hasse diagram. The answer is that we can always reconstruct a flow network from \emph{any Hasse diagram that has the least and greatest elements} following the steps below:
\begin{enumerate}
    \item Identify the least element of the given Hasse diagram as the source and the greatest element as the sink.
    \item Choose a set of paths from the source to the sink via vertices in a non-decreasing order such that for each edge in the given Hasse diagram, it belongs to some paths in the chosen set. 
    \item Assign each edge a capacity that is equal to the number of times it appears in the chosen set of paths.
\end{enumerate}
\begin{figure}[h]
    \centering
    \includegraphics[scale=1.15]{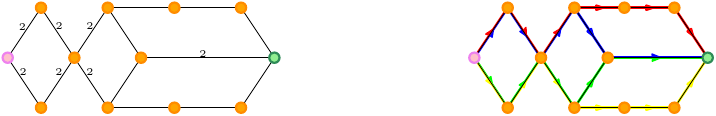}
    \caption{The reconstruction of a flow network from the Hasse diagram in \cref{fig:Hasse diagram}. Left: the flow network; Right: a max-flow on the network. }
    \label{fig:reconstruction of a flow network}
\end{figure}
Then the resulting flow network can host a flow that is equal to the capacity of edge everywhere in the network, and by \cref{def:transitive closure}, the partial order induced by this flow is represented precisely by the given Hasse diagram. Similar questions have also been studied in \cite{escalante_schnittverbande_1972} and \cite{meyer_lattices_1982}, where the authors have noticed that the (vertex or edge) cut sets for a flow network form a lattice\footnote{A lattice is a poset $(X, \leq)$ such that for each pair of elements $x$ and $y$ from $X$, there is a unique greatest lower bound $x\wedge y\in X$ such that $x\wedge y\leq x$ and $x\wedge y\leq y$ and a unique least upper bound $x\vee y\in X$ such that $x\leq x\vee y$ and $y\leq x\vee y$.} and the minimum (vertex or edge) cut sets form a distributive sublattice and they consider the question whether it is possible to reconstruct a flow network with its cut sets, minimum cut sets, or both consistent with a given lattice, a distributive lattice, or a lattice with a distributive sublattice. The reconstruction of a flow network is possible for a few cases using the method in \cite{escalante_schnittverbande_1972} while it remains an open question whether there exists a flow network with its minimum edge-cut sets forming an arbitrarily given distributive lattice. We emphasize that this open question is equivalent to ask whether the closure of the poset $(\{S_x\}_{x\in V(G)}, \subseteq)$ under finite unions and intersections is an arbitrary distributive lattice or not in our language, and the reconstruction proposed in this subsection might be a key to the open question. 

\subsection{Series-parallel graphs} 
\label{sec:series-parallel}
There are several properties about a poset that can be read from its Hasse diagram. For example, if the Hasse diagram is ``N''-free (see \cref{fig:lattice and N-free}), 
\begin{figure}[htb]
    \centering
    \includegraphics[scale=1.15]{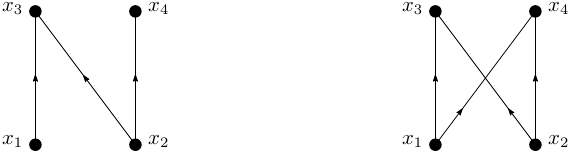}
    \caption{The ``N'' directed graph (left) and the ``double-N'' directed graph (right). They are respectively not allowed as a subgraph in the Hasse diagram of a poset with a series-parallel partial order and a lattice.}
    \label{fig:lattice and N-free}
\end{figure}
meaning that there does not exist four elements $x_1, x_2, x_3, x_4$ in the poset such that $x_1\leq x_3$, $x_2\leq x_4$, and $x_2\leq x_3$ while $x_1$ and $x_4$ are incomparable, then the represented poset has a series-parallel order, which is a partial order that can be constructed from atomic elements using two composition operations \cite{valdes_recognition_1979, goos_complete_1997}. If the Hasse diagram contains a ``double-N'' directed subgraph (see \cref{fig:lattice and N-free}), then the represented poset cannot be a lattice. What interests us in this subsection is the Hasse diagram for a \emph{finite} lattice that has a series-parallel partial order. In this case, the Hasse diagram can be generated recursively by two composition operations --- the series and parallel compositions \cite{eppstein_parallel_1992}. Here, we review the two operations (see \cite[Figure 6]{fitter_max-flow_2024} for visualization).

\begin{definition}[Series-parallel compositions]
    Given two graphs $H_1$ and, resp. $H_2$ each with two distinguished vertices, the source $s_1$ (resp. $s_2$) and the sink $t_1$ (resp. $t_2$), 
    \begin{itemize}
        \item the \emph{series composition} of $H_1$ and $H_2$, denoted by $H_1\seriescon H_2$, is obtained by taking the disjoint union of the two graphs under the condition that $t_1=s_2$; 
        \item the \emph{parallel composition} of $H_1$ and $H_2$, denoted by $H_1\parallelcon H_2$, is obtained by taking the disjoint union of the two graphs under the condition that $s_1=s_2$ and $t_1=t_2$.
    \end{itemize}
\end{definition}

\noindent The Hasse diagram thus obtained recursively from composing the atomic (directed) graphs $G_{\text{triv}}=(\{s, t\}, \{(s, t)\})$ is a series-parallel (directed) graph. Note that our notion of series-parallel graphs is different from that in \cite{valdes_recognition_1979, goos_complete_1997}, where they call the Hasse diagram for representing a series-parallel partial order a series-parallel digraph while they do not require that the poset is necessarily a lattice. The notion of series-parallel graphs that we consider here follows from \cite{eppstein_parallel_1992}; if the edges are directed, the graphs belong to a subset of the class of series-parallel digraphs in \cite{valdes_recognition_1979, goos_complete_1997}. With that said, any property that holds for the series-parallel digraphs in \cite{valdes_recognition_1979, goos_complete_1997} should also hold for the series-parallel graphs here. In particular, one can determine whether a graph is series-parallel in linear time and hence the computational complexity for questions on the series-parallel graphs is always simple. Later in \cref{sec:graph-dependent measures}, we will show that it is possible to define a measure on this simple class of series-parallel graphs, which has important physical applications. 

\section{Random tensor networks}
\label{sec:random-tensor-networks}

Tensor networks are an important variational class of wave functions to approximate ground states in many-body physics \cite{cirac_matrix_2021}. Random tensor networks correspond to probability distributions on the set of tensor networks; equivalently, they are ensembles of tensor networks. For example, they can be constructed from the contraction of independent random tensors placed on a network, and their entanglement properties are of particular interest as they are related to the min-cuts of the network \cite{hayden_holographic_2016, cheng_random_2022}. In this section, we will explore more about the connections to graphs and combinatorics for a random tensor network. 

\subsection{Definition}

The construction of a random tensor network involves a graph $\mathcal{G}=(\mathcal{V}, \mathcal{E})$, which represents the underlying tensor product structure of the random tensor we want to build. Here, we consider the graph $\mathcal{G}$ to be an open graph, meaning, a graph with internal edges (bulk edges) and external half edges (boundary edges). Then, $\mathcal{E} = \mathcal{E}_b \cup \mathcal{E}_\partial$ where $\mathcal{E}_b$ and $\mathcal{E}_{\partial}$ denote the set of bulk and boundary edges respectively and  
\begin{align*}
    \mathcal{E}_b &= \{e_{xy}\in \mathcal{E}|x, y\in V\}, \\ 
    \mathcal{E}_\partial &= \{e_{x}\in \mathcal{E}|x\in V\}.
\end{align*}
Such a graph can be used to define a random tensor network if we associate each external half edge $e_x$ with a $D_{e_x}$ dimensional complex vector space $\mathbb{C}^{D_{e_x}}$ and each internal edge $e_{xy}$ with a product of two $D_{e_{xy}}$ dimensional complex vector spaces $\mathbb{C}^{D_{e_{xy}}}\otimes\mathbb{C}^{D_{e_{xy}}}$ (an internal edge can be regarded as a pairing of two internal half edges). Then, for each vertex $x\in \mathcal{V}$, we use $\mathcal{H}_x$ to denote the product of the Hilbert spaces associated to those edges which are incident to the vertex, i.e. 
\begin{equation*}
    \mathcal{H}_x\coloneqq\left(\bigotimes_{y:e_{xy}\in \mathcal{E}_b}\mathbb{C}^{D_{e_{xy}}}\right)\otimes \left(\bigotimes_{e_{x}\in \mathcal{E}_{\partial}}\mathbb{C}^{D_{e_{x}}}\right)\qquad \forall x\in \mathcal{V},
\end{equation*}
and similarly, for each internal edge $e\in \mathcal{E}_b$, we have 
\begin{equation*}
    \mathcal{H}_{e}\coloneqq\mathbb{C}^{D_{e}}\otimes\mathbb{C}^{D_{e}}\qquad\quad \forall e\in \mathcal{E}_b.
\end{equation*}
We will pick an independent random complex standard Gaussian tensor $g_{i_1 \cdots i_{\deg(x)}}$ at each vertex $x$ and a normalized maximally entangled pair state $|\Omega_e\rangle$ from the Hilbert space at each internal edge $e$:  
\begin{align*}
    \forall x\in \mathcal{V}, &\quad x \to |g_x\rangle\in\mathcal{H}_x\quad\text{with $|g_x\rangle=\sum_{i_1 \cdots i_{\deg(x)}}g_{i_1 \cdots i_{\deg(x)}}|i_1\cdots i_{\deg(x)}\rangle_x$ (not normalized)}, \\
    \forall e\in \mathcal{E}_b, &\quad e \to |\Omega_e\rangle\in\mathcal{H}_e\quad\text{with $|\Omega_e\rangle=\frac{1}{\sqrt{D_e}}\sum_{i=1}^{D_e}|i\rangle_x|i\rangle_y$ ($e$ connects $x$ with $y$)},
\end{align*}
where $\{|i\rangle_x\}$ is an orthonormal basis we choose for the Hilbert space associated with each (external or internal) half edge incident with the vertex $x$. We recall that a complex standard Gaussian random variable is defined as $Z=\frac{1}{\sqrt{2}}(X+iY)$ where $X$ and $Y$ are independent standard real Gaussian (normal) random variables. A random complex standard Gaussian vector (or tensor) is a vector (resp. tensor) $\boldsymbol{Z}=(Z_i)_{i\in I}$ where the entries $\{Z_i\}_{i\in I}$ are independent and identically distributed random variables having a complex standard Gaussian distribution. We then define random tensor networks as follows. 

\begin{definition}
    \label{def:random tensor networks}
    Given an open graph $\mathcal{G}=(\mathcal{V}, \mathcal{E}_b\cup \mathcal{E}_\partial)$, a random tensor network is defined as the projection of the product state of independent random complex standard Gaussian vector assigned to each vertex to the maximally entangled pair state assigned to each bulk edge: 
    \begin{equation*}
        \lvert\psi_{\mathcal{G}}\rangle \coloneqq \left\langle\bigotimes_{e\in \mathcal{E}_b}\Omega_e\right\rvert\left.\bigotimes_{x\in \mathcal{V}}g_x\right\rangle\in \bigotimes_{e\in \mathcal{E}_\partial}\mathbb{C}^{D_{e}}\eqqcolon\mathcal{H}_{\partial}.
    \end{equation*}
    Note that $|\psi_{\mathcal{G}}\rangle$ is not normalized and it is a random state that lives in the Hilbert space at the boundary of the open graph. 
\end{definition}
\subsection{Bipartite entanglement}
One of the most important properties for a multiparty quantum state, which includes tensor networks, is its entanglement \cite{horodecki_quantum_2009, eisert_area_2010}. Here, we shall give a review focusing on the most well understood case --- bipartite entanglement. At first, we need to embed a random tensor network into a bipartite system. Suppose $\{\mathcal{E}_A, \mathcal{E}_B\}$ is a bipartition of the set of boundary edges $\mathcal{E}_\partial$. Accordingly we can decompose the Hilbert space on the boundary as the tensor product $\mathcal{H}_{\partial}=\mathcal{H}_{A}\otimes\mathcal{H}_{B}$ of the following two subspaces
\begin{equation*}
    \mathcal{H}_A\coloneqq\bigotimes_{e\in\mathcal{E}_A}\mathbb{C}^{D_{e}}\quad\text{and}\quad \mathcal{H}_B\coloneqq\bigotimes_{e\in\mathcal{E}_B}\mathbb{C}^{D_e}.
\end{equation*}

\begin{figure}[htb]
    \centering
    \includegraphics[scale=1.15]{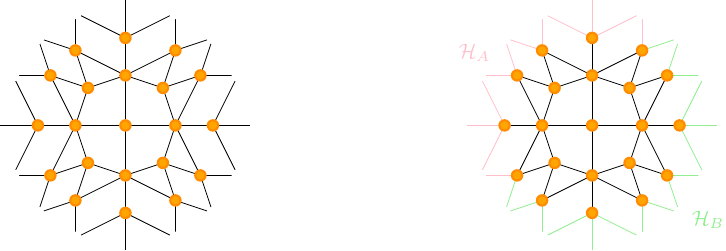}
    \caption{Random tensor networks and the bipartition of the boundary edges.}
    \label{fig:the bipartition of the boundary edges}
\end{figure}

\noindent From now on, the random tensor network $|\psi_{\mathcal{G}}\rangle$ will be regarded as an unnormalized pure state on a two-party system $A-B$. 
\begin{equation*}
    |\psi_{\mathcal{G}}\rangle \in \mathcal{H}_A\otimes\mathcal{H}_B.
\end{equation*}
A pure state $|\psi\rangle_{AB}$ on a two-party system is entangled if and only if it cannot be decomposed as a tensor product $|\psi\rangle_A\otimes|\psi\rangle_{B}$; equivalently, $|\psi\rangle_{AB}$ is entangled if and only if its reduced state on one subsystem (e.g. $\rho_A\coloneqq \operatorname{Tr}_{B}|\psi\rangle_{AB}\langle\psi|$) is mixed. Hence, to study bipartite entanglement of pure state, it suffices to examine the spectrum of the reduced density matrix, where the latter can be obtained by taking the partial trace on the subspace $\mathcal{H}_A$. Here, to study the entanglement property of $|\psi_{\mathcal{G}}\rangle$, we need to look at the spectrum of the following reduced matrix   
\begin{equation*}
    \rho_B = \operatorname{Tr}_A(|\psi_{\mathcal{G}}\rangle\langle\psi_{\mathcal{G}}|).
\end{equation*}
The above reduced matrix is random and depends, of course, on the choice of the bipartition of the set of boundary edges. To characterize the spectrum of $\rho_B$, the standard approach in random matrix theory is to compute the moments, $\mathbb{E}\operatorname{Tr}(\rho_B^n)$ for $n$ an integer. Indeed, it suffices to have a finite number ($=\dim\mathcal{H}_B$) of these moments computed. We will not repeat the calculations of these moments as there are many examples of calculations in \cite{hayden_holographic_2016, cheng_random_2022, fitter_max-flow_2024}. Nevertheless, we point out that, for computing the $n$-th moment, one needs to introduce a map from the set of vertices to the symmetry group on $n$ elements, 
\begin{equation}
    \label{eqn:spin model}
    \begin{split}
        \boldsymbol{\alpha}: \mathcal{V} &\to \mathcal{S}_n \\
        x &\mapsto \alpha_x, 
    \end{split}
\end{equation}
so that as a function of $\boldsymbol{\alpha}$, the $n$-th moment contains a sum over all of the possible choices for such an $\boldsymbol{\alpha}$. 
\begin{equation*}
    \mathbb{E}\operatorname{Tr}(\rho_B^n)=\frac{1}{\prod_{e\in\mathcal{E}_b}D_e^n}\sum_{\{\alpha_x\in\mathcal{S}_n\}_{x\in \mathcal{V}}}\prod_{e_x\in\mathcal{E}_A}D_{e_x}^{\#\alpha_x}\prod_{e_{xy}\in\mathcal{E}_b}D_{e_{xy}}^{\#(\alpha_x^{-1}\alpha_y)}\prod_{e_x\in\mathcal{E}_B}D_{e_x}^{\#(\alpha_x^{-1}\gamma)},
\end{equation*}
where $\#\alpha_x$ is the number of cycles of the permutation $\alpha_x$ and $\gamma=(1\cdots n)\in\mathcal{S}_n$. Note that for $n=1$, the formula above reads $\mathbb{E}\operatorname{Tr}\rho_B=\prod_{e\in\mathcal{E}_{\partial}}D_e$. We introduce now a quasi-normalized version of $\rho_B$, $\Tilde{\rho}_B\coloneqq\frac{\rho_B}{\prod_{e\in\mathcal{E}_{\partial}}D_e}$ (so that $\mathbb{E}\operatorname{Tr}\Tilde{\rho}_B=1$). The $n$-th moment of $\Tilde{\rho}_B$ is then 
\begin{equation}
    \label{eqn:n-th moment}
    \mathbb{E}\operatorname{Tr}(\Tilde{\rho}_B^n)=\sum_{\{\alpha_x\in\mathcal{S}_n\}_{x\in \mathcal{V}}}\prod_{e_x\in\mathcal{E}_A}D_{e_x}^{-|\alpha_x|}\prod_{e_{xy}\in\mathcal{E}_b}D_{e_{xy}}^{-|\alpha_x^{-1}\alpha_y|}\prod_{e_x\in\mathcal{E}_B}D_{e_x}^{-|\alpha_x^{-1}\gamma|},
\end{equation}
where $|\alpha_x|$ counts the minimal number of transpositions needed to write $\alpha_x$ as a product of transpositions and is alternatively defined as 
\begin{equation*}
    |\alpha_x| \coloneqq n - \#\alpha_x. 
\end{equation*}
Indeed, $(\alpha_x, \alpha_y)\mapsto|\alpha_x^{-1}\alpha_y|$ defines a metric on the set $\mathcal{S}_n$, which is known as \emph{Cayley distance} between the two permutations $\alpha_x$ and $\alpha_y$. As expected, Cayley distance is \cite[Proposition 23.20]{nica_lectures_2006}
\begin{itemize}
    \item non-negative, i.e. $|\alpha_x^{-1}\alpha_y|\geq 0$ with the inequality saturated if and only if $\alpha_x=\alpha_y$,
    \item symmetric, i.e. $|\alpha_{x}^{-1}\alpha_y|=|\alpha_y^{-1}\alpha_x|$,
    \item and it satisfies the triangle inequality, i.e. $|\alpha_{x_1}^{-1}\alpha_{x_2}| + |\alpha_{x_2}^{-1}\alpha_{x_3}|\geq|\alpha_{x_1}^{-1}\alpha_{x_3}|.$
\end{itemize}
\begin{figure}[htb]
    \centering
    \includegraphics[scale=1.15]{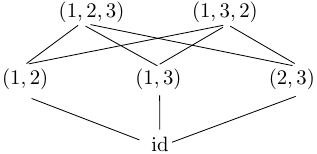}
    \caption{The Cayley graph of $\mathcal{S}_3$. Each edge represents a distance of 1.}
    \label{fig:Cayley graph}
\end{figure}

We emphasize that the quasi-normalized reduced matrix $\Tilde{\rho}_B$ should be regarded as a good approximation of the reduced density matrix $\rho_B/\operatorname{Tr}\rho_B$. This is because when $D_e$ is large for each edge, the distribution of $\operatorname{Tr}\Tilde{\rho}_B$ is centered at $1$ and has a small variance \cite{fitter_max-flow_2024}, and hence the $n$-th R\'enyi entropy, which characterizes the bipartite entanglement on the random tensor network, is approximately proportional to the logarithm of the trace of the $n$-th power of $\Tilde{\rho}_B$, that is,  
\begin{equation}
    \label{eqn:approximated entropy}
    \mathsf{H}_n \coloneqq \frac{1}{1-n}\log \operatorname{Tr}\left[\left(\frac{\rho_B}{\operatorname{Tr}\rho_B}\right)^n\right]\approx\frac{1}{1-n}\log\operatorname{Tr}(\Tilde{\rho}_B^n)
\end{equation}
More importantly, the $n$-th R\'enyi entropy also has a typical value when $D_e$ is large for each edge, which can be well approximated by the logarithm of the $n$-th moment of $\Tilde{\rho}_B$ \cite{hayden_holographic_2016, fitter_max-flow_2024}, that is, 
\begin{equation}
    \label{eqn:approximated average entropy}
        \mathbb{E}\mathsf{H}_n \approx \frac{1}{1-n}\mathbb{E}\log\operatorname{Tr}(\Tilde{\rho}_B^n) \approx \frac{1}{1-n}\log\mathbb{E}\operatorname{Tr}(\Tilde{\rho}_B^n).
\end{equation} 
It then suffices to compute the dominant contribution to the moments of $\Tilde{\rho}_B$ for determining the typical behavior of the bipartite entanglement on a random tensor network. 

\subsection{The asymptotic entanglement entropy}
In the case where the edge dimension $\{D_e\}$ approach infinity, the asymptotic behavior of the entanglement entropy of a random tensor network has been shown to be deeply connected to the min-cuts of the network \cite{hayden_holographic_2016, cheng_random_2022}. Here, we will introduce this relation as well as the asymptotics of the moments systematically. In particular, we consider the following scaling 
\begin{equation*}
    \forall e\in\mathcal{E}\quad D_e = r_e D^{c_e}\quad \text{with $r_e, c_e>0$ fixed and $D\to\infty$}, 
\end{equation*}
with two finite parameters $r_e$ and $c_e$ controlling the speed of $D_e$ approaching infinity as $D$ approaches infinity. For simplifying the discussions, we consider a slightly modified graph that is defined below. 
\begin{definition}
    \label{def:modified graph}
    Given an open graph $\mathcal{G}=(\mathcal{V}, \mathcal{E}_b\cup \mathcal{E}_\partial)$ and a bipartition $\{\mathcal{E}_A, \mathcal{E}_B\}$ of $\mathcal{E}_\partial$, the graph $G_{A-B}=(V, E)$ is obtained by inserting two vertices $A$ and $B$ into $\mathcal{G}$ and connecting $A$ (respectively, $B$) with each half edge in $\mathcal{E}_A$ (respectively, $\mathcal{E}_B$). The set of vertices and the set of edges of the graph $G_{A-B}$ are given as follows: 
    \begin{equation*}
        V\coloneqq \mathcal{V}\cup\{A, B\}, \quad E\coloneqq \mathcal{E}_b\cup \bigcup_{e_x\in \mathcal{E}_A}\{e_{Ax}\} \cup \bigcup_{e_x\in \mathcal{E}_B}\{e_{xB}\}.
    \end{equation*}
\end{definition}
\begin{remark}
    Note that the edges in $E$ are also bidirected where $(x, y)$ and $(y, x)$ are the two directions for $e_{xy}\in E$. The two appended vertices $A$ and $B$ should be understood as the source and sink of the flow network $G_{A-B}$ where the capacity for each edge $e\in E$ is $c_e=\lim_{D\to\infty}\log D_e/\log D$. As an example, $G_{A-B}$ for the network in \cref{fig:the bipartition of the boundary edges} is the same as in \cref{fig:example of a flow network}.
\end{remark}
In fact, there is a natural extension of $\boldsymbol{\alpha}$ from \cref{eqn:spin model} for the graph $G_{A-B}$. That is, there is a map $\boldsymbol{\beta}: V\to\mathcal{S}_n$ such that 
\begin{equation}
    \label{eqn:extended spin model}
    \boldsymbol{\beta}: x\mapsto \beta_x=\left\{\begin{array}{cl}
            \mathrm{id} & \text{if $x=A$} \\
            \alpha_x & \text{if $x\in\mathcal{V}$} \\
            \gamma & \text{if $x=B$}
        \end{array}\right., 
\end{equation}
where $\mathrm{id}$ is the identity in $\mathcal{S}_n$ and $\gamma=(1\cdots n)\in\mathcal{S}_n$. We can also view the extended map $\boldsymbol{\beta}$ as a spin model on the graph $G_{A-B}$ where at each vertex $x$, it is assigned with a spin variable $\beta_x$ and the Hamiltonian of this spin model is given by 
\begin{equation}
    \label{eqn:Hamiltonian}
    H_{G_{A-B}}(\boldsymbol{\beta})=\sum_{e_{xy}\in E}c_{e_{xy}}|\beta_x^{-1}\beta_y|, 
\end{equation}
which shall be understood as the sum of all of the interactions between two neighboring spins. It is not hard to see that the $n$-th moment $\mathbb{E}\operatorname{Tr}(\Tilde{\rho}_B^n)$ from \cref{eqn:n-th moment} is a sum over all of the possible energy configurations for the spin model $\boldsymbol{\beta}$ satisfying the boundary conditions $\beta_A=\mathrm{id}$ and $\beta_B=\gamma$. That is, 
\begin{equation*}
    \mathbb{E}\operatorname{Tr}(\Tilde{\rho}_B^n)=\sum_{\{\beta_x\in \mathcal{S}_n\}_{x\in\mathcal{V}}}D^{-H_{G_{A-B}}(\boldsymbol{\beta})}\prod_{e_{xy}\in E}r_{e_{xy}}^{-|\beta_x^{-1}\beta_y|}.
\end{equation*}
Thus, the random tensor network and its related problems defined on the graph $\mathcal{G}$ and $G_{A-B}$ are essentially equivalent and we shall not always clearly distinguish the two pictures. Nevertheless, the cuts concerned with the asymptotic behaviors of the moments have a much simpler interpretation as a bipartition of vertices in the graph $G_{A-B}$ as from \cref{def:cuts}. Among these cuts, the minimum $A-B$ cuts give a tight lower bound on the Hamiltonian of interest, that is, 
\begin{equation}
    \label{eqn:lower bound on the Hamiltonian}
    \min_{\boldsymbol{\beta}} H_{G_{A-B}}(\boldsymbol{\beta}) = (n-1)\mincut(G_{A-B}),
\end{equation}
and subsequently, we can find the asymptotic expansion of the $n$-th moment as 
\begin{multline}
    \label{eqn:asymptotic expansion of moments}
    \mathbb{E}\operatorname{Tr}(\Tilde{\rho}_B^n) = D^{-\min_{\boldsymbol{\beta}}H_{G_{A-B}}(\boldsymbol{\beta})}(m_n + O(D^{-1})), \\
    \text{where } m_n \coloneqq \lim_{D\to\infty}D^{-\mincut(G_{A-B})}\mathbb{E}\operatorname{Tr}\left[\left(D^{\mincut(G_{A-B})}\Tilde{\rho}_B\right)^n\right]. 
\end{multline}
Combined with \crefrange{eqn:approximated entropy}{eqn:approximated average entropy} where the approximation becomes exact in the limit $D\to\infty$ \cite{fitter_max-flow_2024}, it follows easily that the entanglement entropy is proportional to the minimal capacity of $A-B$ cuts asymptotically. 
\begin{equation*}
    \mathsf{H}_n \xrightarrow{D\to\infty} \log D\cdot\mincut(G_{A-B}),
\end{equation*}
which is one of the important results in \cite{hayden_holographic_2016}. Moreover, the $n$-th properly rescaled moment $m_n$, which gives a finite correction to the entanglement entropy, also depends nontrivially on the min-cuts, and it will be discussed in the next subsection.

\subsection{Order morphisms}
\label{sec:order morphisms}
Indeed, the family of min-cuts contains information more than the lower bound for the Hamiltonian of the spin model on $G_{A-B}$. One application of the canonical partial order defined from the min-cuts is that it gives the number of energy configurations such that the Hamiltonian is minimized and hence the $n$-th properly rescaled moment $m_n$. To be brief, the Hamiltonian of the spin model on $G_{A-B}$ is minimized only if the map $\boldsymbol{\beta}$ sends the partial order $\leq_{\mincut}$ on $V$ to the partial order $\leq_{\mathrm{nc}}$ on $\mathcal{S}_n$, where $\alpha_1\leq_{\mathrm{nc}}\alpha_2$ means that the cycles of $\alpha_1$ form a noncrossing partition for the cycles of $\alpha_2$ \cite[Proposition 23.23]{nica_lectures_2006} (also see \cref{appendix:free-probability} and \cref{fig:crossing partition and non-crossing partition}).

\begin{figure}[htb]
    \centering
    \includegraphics[scale=1.15]{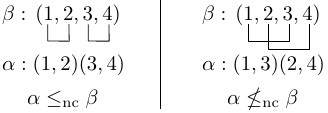}
    \caption{Crossing partition and non-crossing partition.}
    \label{fig:crossing partition and non-crossing partition}
\end{figure}
\noindent The following theorem can then be proved using \cref{theorem:equivalence of all of the partial orderings}. 

\begin{theorem}
    \label{theorem:order morphism}
    The Hamiltonian $H_{G_{A-B}}(\boldsymbol{\beta})$ from \cref{eqn:Hamiltonian} is minimized if and only if $\boldsymbol{\beta}$ is an element of the following set 
    \begin{equation*}
        \mathrm{Ord}(V, \mathcal{S}_n) \coloneqq \{\boldsymbol{\beta}:(V, \leq_{\mincut})\to(\mathcal{S}_n, \leq_{\mathrm{nc}})|\beta_A=\mathrm{id}\ \text{and\ } \beta_B=\gamma=(1\cdots n)\},
    \end{equation*}
    that is, $\boldsymbol{\beta}$ is an order morphism from the poset $(V, \leq_{\mincut})$ to the poset $(\mathcal{S}_n, \leq_{\mathrm{nc}})$, i.e. $x\leq_{\mincut} y\implies \beta_x\leq_{\mathrm{nc}}\beta_y$, with the boundary conditions $\beta_A=\mathrm{id}$ and $\beta_B=\gamma$. Let $w_{G_{A-B}}(\boldsymbol{\beta})$ be the weight of a map $\boldsymbol{\beta}$ defined as follows
    \begin{equation*}
        w_{G_{A-B}}(\boldsymbol{\beta}) \coloneqq \prod_{e_{xy}\in E}r_{e_{xy}}^{-|\beta_x^{-1}\beta_y|}
    \end{equation*}
    Then the $n$-th properly rescaled moment $m_n$ from \cref{eqn:asymptotic expansion of moments} is given by 
    \begin{equation*}
        m_n=\sum_{\boldsymbol{\beta}\in\mathrm{Ord}(V, \mathcal{S}_n)}w_{G_{A-B}}(\boldsymbol{\beta}).
    \end{equation*} 
\end{theorem}
\begin{remark}
    The theorem gives a combinatorial interpretation for the finite correction to the entanglement entropy $m_n$. Meanwhile, the set $\{m_n\}_{n\in\mathbb{N}}$ is also related to the moments of a graph-dependent measure, which is discussed in the next section. 
\end{remark}
\begin{proof}
    Here, we consider $P$ as a path from $A$ to $B$ that one can find from the graph $G_{A-B}$, 
    \begin{equation*}
        P\coloneqq\{(x_i, x_{i+1})\in E\}_{i\in\{0, \cdots, |P|\},\  x_0=A,\ \text{and}\ x_{|P|}=B}, 
    \end{equation*}
    and $\mathcal{P}_{A-B}$ as a set of edge-disjoint paths from $A$ to $B$,   
    \begin{equation*}
        \mathcal{P}_{A-B}\coloneqq\{P|\forall P_1, P_2\in\mathcal{P}_{A-B}\ P_1\cap P_2=\emptyset\}.
    \end{equation*}
    The optimal amount of flow on a set of edge-disjoint paths from $A$ to $B$ is then defined as  
    \begin{equation*}
        |\mathcal{P}_{A-B}|_c \coloneqq \sum_{P\in\mathcal{P}_{A-B}}\min_{e\in P}c_e.
    \end{equation*}
    This definition is reasonable since we can always find a flow $f$ such that $|f|_c=|\mathcal{P}_{A-B}|_c$ by letting $f(x_i, x_{i+1})=\min_{e\in P}c_e$ for all $(x_i, x_{i+1})\in P\in\mathcal{P}_{A-B}$ and $f(x, y)=0$ for the other edges. Then if $|\mathcal{P}_{A-B}|_c$ is maximized, $\mathcal{P}_{A-B}$ corresponds to a max-flow $f$, and \cref{eqn:lower bound on the Hamiltonian} implies that
    \begin{equation*}
        \min_{\boldsymbol{\beta}} H_{G_{A-B}}(\boldsymbol{\beta})=(n-1)\max_{\mathcal{P}_{A-B}}|\mathcal{P}_{A-B}|_c=|\mathrm{id}^{-1}\gamma|\max_{\mathcal{P}_{A-B}}|\mathcal{P}_{A-B}|_c.
    \end{equation*}
    Now suppose that $\boldsymbol{\beta}$ is an energy configuration such that the Hamiltonian is minimized. Then for a set of edge-disjoint paths $\mathcal{P}_{A-B}$ with a maximal optimal amount of flow, we have 
    \begin{equation}
        \label{eqn:minimization of the Hamiltonian}
        \begin{split}
            &H_{G_{A-B}}(\boldsymbol{\beta})-|\mathrm{id}^{-1}\gamma|\cdot|\mathcal{P}_{A-B}|_c\\
            =&\underbrace{H_{G_{A-B}}(\boldsymbol{\beta})-\sum_{P\in\mathcal{P}_{A-B}}\min_{e\in P}c_e\sum_{(x, y)\in P}|\beta_x^{-1}\beta_y|}_{\Delta_1}+\underbrace{\sum_{P\in\mathcal{P}_{A-B}}\min_{e\in P}c_e\left(\sum_{(x, y)\in P}|\beta_x^{-1}\beta_y|-|\mathrm{id}^{-1}\gamma|\right)}_{\Delta_2} \\
            =&\ 0
        \end{split}
    \end{equation}
    Note that $\Delta_1$ is non-negative since it either contains the non-negative term $c_{e_{xy}}|\beta_x^{-1}\beta_y|$ for each edge $(x, y)$ that does not belong to any of the edge-disjoint paths from $A$ to $B$ or the non-negative term $(c_{e_{xy}}-\min_{e\in P}c_e)|\beta_x^{-1}\beta_y|$ for each edge $(x, y)$ that belongs to some edge-disjoint paths $P$ from $A$ to $B$ (recall \crefrange{eqn:extended spin model}{eqn:Hamiltonian}) and $\Delta_2$ is non-negative since the triangle inequality,  
    \begin{equation*}
        |\alpha_1^{-1}\alpha_2|+|\alpha_2^{-1}\alpha_3|\geq|\alpha_1^{-1}\alpha_3|,
    \end{equation*}
    holds for Cayley distance. Then \cref{eqn:minimization of the Hamiltonian} is true if and only if $\Delta_1$ and $\Delta_2$ vanishes simultaneously, which is then equivalent to saying that 
    \begin{itemize}
        \item $\beta_x=\beta_y$ for each edge $(x, y)$ that does not belong to any of the edge-disjoint paths from $A$ to $B$ (hence Cayley distance $|\beta_x^{-1}\beta_y|$ vanishes),
        \item $\beta_x=\beta_y$ for each edge $(x, y)$ that belongs to some edge-disjoint paths $P$ but with its capacity greater than the minimal edge capacity along the path, i.e. $c_{e_{xy}}> \min_{e\in P}c_e$ (hence Cayley distance $|\beta_x^{-1}\beta_y|$ vanishes again),
        \item $\mathrm{id}\leq_{\mathrm{nc}}\beta_{x_{1}}\leq_{\mathrm{nc}}\cdots\leq_{\mathrm{nc}}\beta_{x_{|P|-1}}\leq_{\mathrm{nc}}\gamma$ for each path $P\in\mathcal{P}_{A-B}$ (hence the triangle inequality saturates by \cite[Notation 23.21]{nica_lectures_2006}),
    \end{itemize}
    Recall that for a set of edge-disjoint paths $\mathcal{P}_{A-B}$ with a maximal optimal amount of flow, we can have a corresponding max-flow $f$ such that $f(x_i, x_{i+1})=\min_{e\in P}c_e$ for all $(x_i, x_{i+1})\in P$ and $f(x, y)=0$ for the other edges. Then the above three conditions are satisfied if and only if $\beta_x=\beta_y$ if either $f(x, y)=0$ when $e_{xy}$ does not belong to any path $P\in\mathcal{P}_{A-B}$ or $f(x, y)=\min_{e\in P}c_e\neq c_{e_{xy}}$ when $e_{xy}$ belongs to some $P\in\mathcal{P}_{A-B}$ and $\beta_x\leq_{\mathrm{nc}}\beta_y$ if $f(x, y)=\min_{e\in P}c_e=c_{e_{xy}}$. By \cref{def:transitive closure} for the partial orders, we can rewrite it as $\beta_x=\beta_y$ if $x=_fy$ and $\beta_x\leq_{\mathrm{nc}}\beta_y$ if $x\leq_f y$. Thus, $H_{G_{A-B}}(\boldsymbol{\beta})$ is minimized if and only if $\boldsymbol{\beta}$ is an order morphism between $(V, \leq_f)$ and $(\mathcal{S}_n, \leq_{\mathrm{nc}})$ with $\beta_A=\mathrm{id}$ and $\beta_B=\gamma$. Note that we have proved that $\leq_f$ is the same as $\leq_{\mincut}$ in \cref{theorem:equivalence of all of the partial orderings}. Then, $\boldsymbol{\beta}$ is also an order morphism between $(V, \leq_{\mincut})$ and $(\mathcal{S}_n, \leq_{\mathrm{nc}})$, which then completes the proof for the first statement in the theorem. The other statement of the theorem directly follows from the observation that 
    \begin{equation*}
        \begin{split}
            m_n&\coloneqq\lim_{D\to\infty} D^{-\mincut(G_{A-B})}\operatorname{Tr}\left[\left(D^{\mincut(G_{A-B})}\Tilde{\rho}_B\right)^n\right] \\
            &= \sum_{\boldsymbol{\beta}\in\operatorname{Ord}(V(G_{A-B}), \mathcal{S}_n)}\prod_{e_{xy}\in E(G_{A-B})}(r_{e_{xy}})^{-|\beta_x^{-1}\beta_y|}.
        \end{split}
    \end{equation*}
\end{proof}
The theorem implies the following result from \cite[Theorem 5.14]{fitter_max-flow_2024}. 
\begin{corollary}
    \label{corollary:order morphisms}
    If $r_e=1$ for all $e\in E$, then
    \begin{equation*}
        m_n=|\mathrm{Ord}(V, \mathcal{S}_n)|.
    \end{equation*}
\end{corollary}
As an example, consider the random tensor network $\mathcal{G}$ in \cref{fig:the bipartition of the boundary edges} with $D_e=D$ for each edge. The slightly modified graph $G_{A-B}$ is given in \cref{fig:example of a flow network} (be careful with the capacity for each edge) with the Hasse diagram for representing the poset $(V, \leq_{\mincut})$ given in \cref{fig:Hasse diagram}. There are in total $2\times 2+1+3\times 2+3\times 3\times 2= 29$ order morphisms (see \cref{fig:order morphisms} below) between $(V, \leq_{\mincut})$ and $(\mathcal{S}_2, \leq_{\mathrm{nc}})$ and hence $m_2=29$ by \cref{corollary:order morphisms}.
\begin{figure}[H]
    \centering
    \includegraphics[scale=1.15]{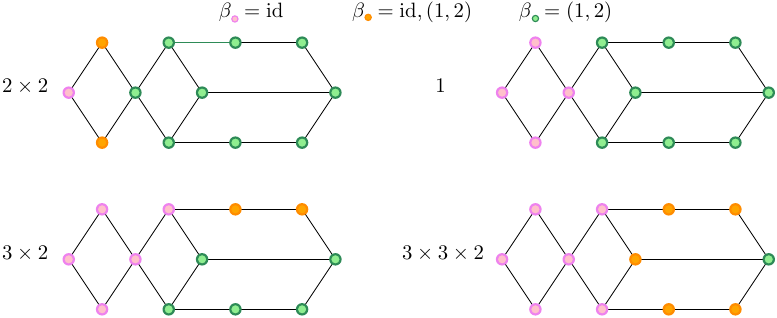}
    \caption{Examples of order morphisms between $(V, \leq_{\mincut})$ and $(\mathcal{S}_2, \leq_{\mathrm{nc}})$ for the Hasse diagram from \cref{fig:Hasse diagram}}
    \label{fig:order morphisms}
\end{figure}

\section{Graph-dependent measures}
\label{sec:graph-dependent measures}
In this section, we discuss the connection between random tensor networks and the free probability theory \cite{voiculescu_free_1992, nica_lectures_2006, mingo_free_2017}. We assume that the readers are familiar with the free probability theory and, in particular, the free additive convolution $\boxplus$ and the free multiplicative convolution $\boxtimes$. These notions are briefly reviewed in \cref{appendix:free-probability} which is strongly recommended for the readers who are not familiar with the free probability theory. The focus of this section is a measure which can be assigned to a series-parallel graph (recall \cref{sec:series-parallel} for series-parallel graphs). Since that graphs have been widely used to represents polynomials of tensors (e.g. random tensor networks), this type of measure corresponds to the limit distribution of some particular random tensors. In the following subsections, we will define and study these measures for series-parallel graphs on which there exists a number $r_e$ for each edge $e$ characterizing the dimension of its associated vector space relatively among all of the edges. These numbers are important since they relatively give an upper bound on the rank of a flattened tensor that could be represented by a graph \cite{cui_quantum_2016} and hence we introduce a tuple $(\mu_G, r_G)$ to track the information contained in a series-parallel graph $G$. 
\subsection{Definition}
Here, we introduce the measure $\mu_G$ and the number $r_G$ associated to a series-parallel graph $G$ more precisely by giving a formal recursive definition. At first, recall that Marchenko-Pastur distribution $\pi_r$ of parameter $r$ (also known as the free Poisson distribution) has the density (also see \cref{fig:MP distribution}) given as follows \cite[Proposition 12.11]{nica_lectures_2006}.
\begin{equation*}
    \pi_{r} = \left\{\begin{array}{ll}
        (1-r)\delta_0 + r\Tilde{\nu} & \text{if $0\leq r\leq 1$} \\
        \Tilde{\nu} & \text{if $r>1$}
    \end{array}\right.
\end{equation*}
where $\Tilde{\nu}$ is the measure supported on $[(1-\sqrt{r})^2, (1+\sqrt{r})^2]$ with density 
\begin{equation*}
    d\Tilde{\nu}(t) = \frac{1}{2\pi  t}\sqrt{4r-(t-(1+r))^2}dt.
\end{equation*} 
This distribution is also the limit $(D\to\infty)$ of the empirical eigenvalue distribution of the Wishart matrix $X^{\dagger}X$ \cite{johnstone_distribution_2001} where $X$ is a rectangular $rD\times D$ matrix with independent identically distributed complex standard random Gaussian entries. Using graphs, these matrices are represented as follows:

\begin{figure}[htb]
    \centering
    \includegraphics[scale=1.15]{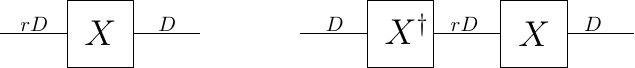}
    \caption{Graphical representation for a Gausisan matrix (left) and a Wishart matrix (right).}
    \label{fig:Wishart matrix}
\end{figure}

\noindent We consider an extension of these rectangular random matrices $X$ by applying the series and parallel compositions introduced in \cref{sec:series-parallel} on independent copies of them. Then the distributions and ranks for the obtained random tensors are related to the tuple $(\mu_G, r_G)$ defined below. We introduce a measure $\mu_G$ associated with any series-parallel graph $G$. This definition is \emph{recursive}. The parallel composition of graphs will correspond to a \emph{classical multiplicative convolution of measures}, reflecting independent systems. The series composition is more intricate, corresponding to a \emph{combination of free multiplicative and additive convolutions} that captures the matrix-like concatenation of random tensors.

\begin{definition}
    \label{def:measure}
    For a series-parallel graph $G$ on which the dimension of the vector space on each edge $e$ is proportional to $r_e$, we associate a tuple $(\mu_G, r_G)$ defined recursively as follows: 
    \begin{itemize}
        \item For the \emph{trivial graph} $G_{triv}=(\{s, t\}, \{e_{st}\})$ on which the dimension of the vector space on the only edge $e_{st}$ is proportional to $r_{e_{st}}$, one has 
        \begin{equation*}
            (\mu_{G_{triv}}, r_{G_{triv}}) \coloneqq (\delta_1, r_{e_{st}})
        \end{equation*}
        \item For the graph $G\seriescon H$ which is the \emph{series composition} of its two subgraphs $G$ and $H$, by letting $r_{\max} = \max\{r_G, r_H\}$ and $r_{\min}=\min\{r_G, r_H\}$, one has 
        \begin{equation*}
            (\mu_{G\seriescon H}, r_{G\seriescon H}) \coloneqq \left(\left(\frac{r_{\min}}{r_G}\odot\mu_G\right)^{\boxplus\frac{r_G}{r_{\min}}}\boxtimes\hat{\pi}_{\frac{r_{\max}}{r_{\min}}}\boxtimes\left(\frac{r_{\min}}{r_H}\odot\mu_H\right)^{\boxplus\frac{r_H}{r_{\min}}}, r_{\min}\right)
        \end{equation*}
        \item For the graph $G\parallelcon H$ which is the \emph{parallel composition} of its two subgraphs $G$ and $H$, one has 
        \begin{equation*}
            (\mu_{G\parallelcon H}, r_{G\parallelcon H}) \coloneqq \left(\mu_G\ast\mu_H, r_G\cdot r_H\right)
        \end{equation*}
    \end{itemize}
    where $\odot$ denotes the rescaling operation (if the random variable $X$ has the distribution $\mu$, then the random variable $rX$ has the distribution $r\odot\mu$), $\hat{\pi}_r$ is the rescaled Marchenko-Pastur distribution having unit mean, that is, $\hat{\pi}_r\coloneqq r^{-1}\odot\pi_r$ (see \cref{fig:MP distribution}), and $\mu_G\ast\mu_H$ is the classical multiplicative convolution. 
\end{definition}

\begin{figure}
    \centering
    \includegraphics[scale=1.15]{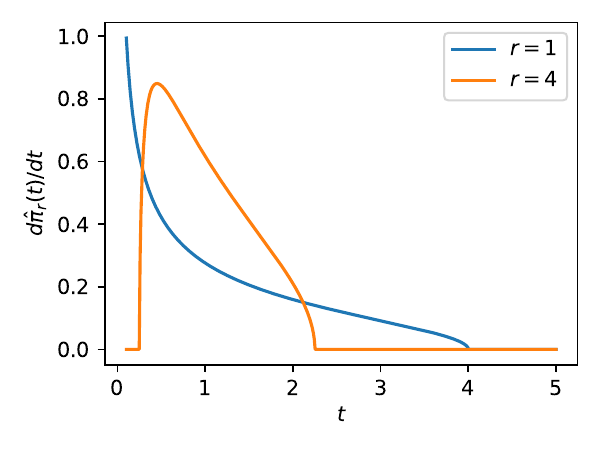}
    \caption{The rescaled Marchenko-Pastur distribution $\hat{\pi}_r=r^{-1}\odot\pi_r$.}
    \label{fig:MP distribution}
\end{figure}

\begin{remark}
    By this definition, we can always associate a tuple $(\mu_G, r_G)$ to a series-parallel graph $G$, as the graph can always be decomposed into a set of trivial graphs, with the tuples $(\mu_{G_{triv}}, r_{G_{triv}})$ for these trivial graphs given, using the reversed series and parallel compositions. The decomposition of a series-parallel graph into trivial graphs is unique up to the associativity of the series composition and the associativity and commutativity of the parallel composition \cite{goos_complete_1997} and hence, the tuple $(\mu_G, r_G)$ for a series-parallel graph $G$ is uniquely defined since one can show that the series compositions and parallel compositions for tuples $(\mu_{H_1}, r_{H_1})$ and $(\mu_{H_2}, r_{H_2})$ are also associative and commutative. The complicated rescaling operation in the series composition $(r^{-1}\odot[\cdot])^{\boxplus r}$ should be understood as projecting the probability space into a smaller space (see \cref{theorem:projection}).
\end{remark}
\begin{remark}
    We also emphasize that the graph-dependent measure $\mu_G$ in our definition is of mean 1 for any series-parallel graph $G$. In fact, the mean of the random tensor represented by the graph $G$ is not necessarily 1 and depends on the relative edge dimensions $\{r_e\}_{e\in E(G)}$. However, by introducing the rescaling parameter $r_G$, it is guaranteed that the mean of the rescaled random tensor is 1, and hence the tensor has the distribution $\mu_G$. Coincidentally, $r_G$ is the asymptotic value of the rank of the random tensor, since according to \cite{hastings_asymptotics_2017}, the rank of a tensor network is asymptotically equal to the minimum value of the product over the edge dimensions of a cut set that separates the source from the sink. 
\end{remark}

In the case where the graph $G$ is a series composition of two trivial graphs, the measure $\mu_G$ is the rescaled Marchenko-Pastur distribution (see \cref{fig:MP distribution}) and the number $r_G$ is the proportion of the largest edge dimension relative to the smallest edge dimension as the dimension goes to infinity. This is consistent with the limit distribution of the (properly rescaled) Wishart matrix $\frac{1}{rD}X^{\dagger}X$ mentioned at the beginning of this section. Hence, it is not hard to see that in general, the measure $\mu_G$ associated with a series-parallel graph $G$ is the empirical eigenvalue distribution of the generalized Wishart matrix that is the product of the random tensor and its adjoint where the former is represented by the graph $G$ with a Gaussian tensor put at each vertex except at $s$ and $t$ (the source and sink of the series-parallel graph). 

\subsection{Measures to path graphs (line graphs)}\label{sec:line-graphs}
In this subsection, we focus on the recursive relation for the series composition appeared in \cref{def:measure}. It is well illustrated by considering the subfamily of measures associated to path graphs (or linear graphs), which are simple yet non-trivial series-parallel graphs such that they are obtained by consecutive series compositions of trivial graphs. The recursive relation is significantly simplified for constructing the measures to path graphs and we can write them out as follows:
\begin{theorem}
    \label{theorem:measures to path graphs}
    If $G$ is a path graph on which the dimension of the vector space on each edge $e$ is proportional to $r_e$, then 
    \begin{equation*}
        \exists e_{\min}\in E(G): r_G=r_{e_{\min}}=\min\{r_e|e\in E(G)\}\quad\text{and}\quad\mu_G = \boxtimes_{e\in E(G)-\{e_{\min}\}}\hat{\pi}_{r_e/r_G}
    \end{equation*}
\end{theorem}
\begin{proof}
    According to \cref{def:measure}, to obtain the measure to the graph which is the series composition of two subgraphs, we need to compute the free additive convolution of the rescaled measures to the subgraphs. A key observation which helps for the computation is that:  
    \begin{equation*}
        (r^{-1}\odot\delta_1)^{\boxplus r} = \delta_1 \quad \forall r\geq 1.
    \end{equation*}
    Hence, if a graph is a series composition of its subgraph and a trival graph, then by definition, we can find that 
    \begin{equation}
        \label{eqn:a recursive rule on the measure}
        \mu_{G\seriescon G_{triv}} = \mu_G\boxtimes\hat{\pi}_{r_{G_{triv}}/r_G} \quad\text{if}\quad r_{G\seriescon G_{triv}}=r_G.
    \end{equation}
    Now we consider the path graphs. By the reversed series compositions, a path graph can be decomposed into a set of trivial graphs on which the dimension of the vector space on the only edge is proportional to $r_e$ coming from the corresponding edge in the path graph. We can then reconstruct the original path graph by applying the series composition recursively starting around the edge $e_1$ of the minimal dimension $r_{e_1}$. During the reconstruction, we apply \cref{def:measure} to obtain a measure and clearly the condition for \cref{eqn:a recursive rule on the measure} holds at each recursive step. This implies that for $G$ a path graph, 
    \begin{equation*}
        r_G=r_{e_1}\quad\text{and}\quad\mu_G = \delta_1\boxtimes\hat{\pi}_{r_{e_2}/r_{e_1}}\boxtimes\cdots\boxtimes \hat{\pi}_{r_{e_{|E(G)|}}/r_{e_1}}
    \end{equation*}
    The expression in the theorem is then simply an equivalent statement as the above.
\end{proof}
\begin{remark} 
    This theorem is also helpful for understanding the complicated rescaling operation $\left(r^{-1}\odot[\cdot]\right)^{\boxplus r}$ in the recursive relation in \cref{def:measure}. For a measure to a path graph $G$, we know from the theorem that $\mu_G=\boxtimes_{e\in E(G)-\{e_{\min}\}}\hat{\pi}_{r_e/r_G}$. Without losing generality, suppose now that we want to construct the measure for the series composition of $G$ and another path graph $H$ and $r_{\min}=r_H<r_G$. By substituting $\mu_G$ with the measure to a path graph, we can find that  
    \begin{equation*}
        \left.\left(\frac{r_H}{r_G}\odot\mu_G\right)^{\boxplus\frac{r_G}{r_H}}\right|_{\mu_G=\boxtimes_{e\in E(G)-\{e_{\min}\}}\hat{\pi}_{r_e/r_G}}=\boxtimes_{e\in E(G)-\{e_{\min}\}}\hat{\pi}_{r_e/r_H}.
    \end{equation*}
    That is, the operation $\left(r^{-1}\odot[\cdot]\right)^{\boxplus r}$ always tunes the relative dimension $r_e$ of an edge $e$ into the ratio of $r_e$ to the minimum among them $\{r_e\}_{e\in E(G\seriescon H)}$ in the whole graph $G\seriescon H$. Hence, it is easy to see that the definition and the theorem are always consistent with each other.
\end{remark}

The measures associated to path graphs are not completely new, and there are several properties that have been studied in the literature. For one thing, the two-parameter family of the \emph{free Bessel distribution} $\pi_{st}$ with $s\in \mathbb{N}$ and $t\geq 1$ \cite{banica_free_2011} is a special case of the measure to a path graph with $\{r_e\}_{e\in E(G)}=\{\underbrace{1, 1, \dots, 1}_{\text{$s$ times}}, t\}$. That is, 
\begin{equation*}
    \pi_{st} = t\odot\mu_{\includegraphics[scale=1.15]{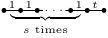}}
\end{equation*}
This distribution has been proposed as a generalization of the Marchenko-Pastur distribution (free Poisson distribution) and it also relates to the Bessel distribution as its classical analog. For another, the measure to a path graph with $\{r_e\}_{e\in E(G)}=\{1, \underbrace{k, k, \dots, k}_{\text{$k$ times}}\}$ converges weakly to a $\boxtimes$-infinitely divisible distribution $\mu$, with $S$ transform $S_{\mu}(z)=\exp(-z)$ (see \cref{appendix:free-probability}), supported on the positive real axis as $k\to\infty$ \cite{anshelevich_local_2014}. The $\boxtimes$-infinitely divisibility of $\mu$ refers to the property that $\forall k\in\mathbb{N}$, there exists a measure $\mu_k$ such that $\mu=\mu_k^{\boxtimes k}$ \cite{bercovici_levy-hinctype_1992}. This property defines a class of measures not only appearing naturally in limit theorems for products of independent random variables in free probability theory, as this class is closed under the free multiplicative convolution $\boxtimes$, but also being related to L\'evy stochastic process. The analytic and combinatorial properties of the free Bessel distribution $\pi_{st}$ and the $\boxtimes$-infinitely divisible distribution $\mu$ have been discussed respectively in \cite{banica_free_2011} and \cite{anshelevich_local_2014}, and we emphasize that the measures of path graphs that we have proposed here interpolate between the two distributions. 

\subsection{Relation to random tensor networks}

We proceed by giving the relation between the graph-dependent measure and random tensor networks discussed in the last section. It then follows from the relation that the combinatorial aspect of the measure has been already presented in \cref{sec:order morphisms}, which we will also summarize here. We start by quickly reviewing random tensor networks. Recall that given an open graph $\mathcal{G}=(\mathcal{V}, \mathcal{E}_b\cup\mathcal{E}_{\partial})$, we can define a random tensor $|\psi_{\mathcal{G}}\rangle\in\bigotimes_{e\in\mathcal{E}_{\partial}}\mathbb{C}^{D_e}$ as in \cref{def:random tensor networks}, where $D_e$ is the dimension of the vector space at the edge $e$ parameterized by $D_e=r_eD^{c_e}$ where $r_e$ resembles those in \cref{def:measure} and $c_e$ resembles those in \cref{def:flows} --- the capacity of the edge. The random tensor can be regarded as an unnormalized entangled state between two subspaces, $\mathcal{H}_A\coloneqq\bigotimes_{e\in\mathcal{E}_{A}}\mathbb{C}^{D_e}$ and $\mathcal{H}_B\coloneqq\bigotimes_{e\in\mathcal{E}_{B}}\mathbb{C}^{D_e}$, once a partition of the set of boundary edges $\mathcal{E}_{\partial}=\mathcal{E}_A\cup\mathcal{E}_B$ is given. To quantify the entanglement of the random tensor network, one can consider the spectra of the reduced matrix or its quasi-normalized version
\begin{equation*}
    \rho_B = \operatorname{Tr}_A(|\psi_{\mathcal{G}}\rangle\langle\psi_{\mathcal{G}}|)\quad\text{and}\quad\Tilde{\rho}_B=\rho_B/\prod_{e\in\mathcal{E}_{\partial}}D_e, 
\end{equation*}
where $\operatorname{Tr}_A$ is the partial trace on the subspace $\mathcal{H}_A$. Note that $\rho_B, \Tilde{\rho}_B$ are $\dim \mathcal{H}_B\times \dim \mathcal{H}_B$ dimensional positive semi-definite random matrices. The distribution of the quasi-normalized matrix is related to the graph-dependent measure according to the following theorem.
\begin{theorem}
    \label{theorem:relation to random tensor networks}
    Let $G_{A-B}$ be the flow network corresponding to $\mathcal{G}$ under a bipartition of the boundary edges $\mathcal{E}_{\partial}=\mathcal{E}_A\cup\mathcal{E}_B$ by \cref{def:modified graph}. If the quotient graph $\vec{G}_{A-B}$ of $G_{A-B}$ defined by \cref{eqn:quotient graph} is series-parallel, then the properly rescaled limit distribution of $\Tilde{\rho}_B$ is the measure to the quotient graph $\mu_{\vec{G}_{A-B}}$ from \cref{def:measure}. That is, 
    \begin{equation}
        \label{eqn:relation between the distribution and moments}
        \varphi_n(\mu_{\vec{G}_{A-B}}) = \lim_{D\to\infty}\frac{1}{r_{\vec{G}_{A-B}}D^{\mincut(G_{A-B})}}\mathbb{E}\operatorname{Tr}\left[(r_{\vec{G}_{A-B}}D^{\mincut(G_{A-B})}\Tilde{\rho}_B)^n\right]
    \end{equation}
    where
    \begin{equation*}
        \varphi_n(\mu)\coloneqq\int_{t\in\mathbb{R}} t^nd\mu.
    \end{equation*}
\end{theorem}
\begin{remark}
    Note that \cref{eqn:relation between the distribution and moments} can also be written as 
    \begin{equation*}
        \varphi_n(\mu_{\vec{G}_{A-B}}) =(r_{\vec{G}_{A-B}})^{n-1}m_n,
    \end{equation*}
    where $m_n$ is defined in \cref{eqn:asymptotic expansion of moments} and also appears in \cref{theorem:order morphism}.
\end{remark}
\begin{proof}
    First, notice that at finite $D$, 
    \begin{equation*}
        \frac{1}{r_{\vec{G}_{A-B}}D^{\mincut(G_{A-B})}}\mathbb{E}\operatorname{Tr}\left[(r_{\vec{G}_{A-B}}D^{\mincut(G_{A-B})}\Tilde{\rho}_B)^n\right]
    \end{equation*}
    corresponds to the $n$-th moment of the empirical eigenvalue distribution of the random matrix $r_{\vec{G}_{A-B}}D^{\mincut(G_{A-B})}\Tilde{\rho}_B$. The limit of these moments satisfies Carleman's condition and hence, according to \cite[Footnote 3]{hastings_asymptotics_2017} and \cite[Theorem 5.14]{fitter_max-flow_2024}, the limit of the empirical eigenvalue distribution of the random matrix $r_{\vec{G}_{A-B}}D^{\mincut(G_{A-B})}\Tilde{\rho}_B$ exists and the moments of this limit distribution are the limits of the moments in the above equation. Here we denote the limit of the $n$-th moment of the empirical eigenvalue distribution by $\Tilde{m}_{n;G_{A-B}}$. According to \cref{theorem:order morphism}, we can find that 
    \begin{equation*}
        \Tilde{m}_{n;G_{A-B}} = (r_{\vec{G}_{A-B}})^{n-1}m_n = \sum_{\boldsymbol{\beta}\in\operatorname{Ord}(V(G_{A-B}), \mathcal{S}_n)}\frac{\prod_{e_{xy}\in E(G_{A-B})}(r_{e_{xy}})^{-|\beta_x^{-1}\beta_y|}}{(r_{\vec{G}_{A-B}})^{1-n}}
    \end{equation*}
    The above expression can be simplified further. Observe that 
    \begin{itemize}
        \item if $e_{xy}$ is an edge between $x, y$ such that $x=_{\mincut}y$, then $\beta_x=\beta_y$ for any $\boldsymbol{\beta}\in\operatorname{Ord}(V(G_{A-B}), \mathcal{S}_n)$ and hence $(r_{e_{xy}})^{-|\beta_x^{-1}\beta_y|}=1$. This is also equivalent to saying that edges which are not in the quotient graph $\vec{G}_{A-B}$ do not contribute to $\Tilde{m}_{n;G_{A-B}}$;
        \item Since the min-cuts in the graph $G_{A-B}$ and the quotient graph $\vec{G}_{A-B}$ are the same, it implies that $\boldsymbol{\beta}$ is an order morphism between $(V(G_{A-B}), \leq_{\mincut})$ and $(\mathcal{S}_n, \leq_{\mathrm{nc}})$ if and only if it corresponds to an order morphism between $(V(\vec{G}_{A-B}), \leq_{\mincut})$ and $(\mathcal{S}_n, \leq_{\mathrm{nc}})$, i.e. 
        \begin{equation*}
            \operatorname{Ord}(V(G_{A-B}), \mathcal{S}_n) \cong \operatorname{Ord}(V(\vec{G}_{A-B}), \mathcal{S}_n).
        \end{equation*}
    \end{itemize}
    We have then proved that the moments of the limit distribution depend only on the quotient graph $\vec{G}_{A-B}$ and we can find that 
    \begin{equation}
        \label{eqn:moments of limit distribution}
        \Tilde{m}_{n;G_{A-B}}=\Tilde{m}_{n;\vec{G}_{A-B}} = \sum_{\boldsymbol{\beta}\in\operatorname{Ord}(V(\vec{G}_{A-B}), \mathcal{S}_n)}\frac{\prod_{(x, y)\in E(\vec{G}_{A-B})}(r_{(x, y)})^{-|\beta_x^{-1}\beta_y|}}{(r_{\vec{G}_{A-B}})^{1-n}}
    \end{equation}
    
    Now we prove $\varphi_n(\mu_{\vec{G}_{A-B}})=\Tilde{m}_{n;\vec{G}_{A-B}}$ by induction. Consider first when the quotient graph $\vec{G}_{A-B}$ is the trivial graph. \cref{eqn:moments of limit distribution} implies that 
    \begin{equation*}
        \Tilde{m}_{n;\vec{G}_{triv}} = \frac{(r_{(s, t)})^{-|\mathrm{id}^{-1}\gamma|}}{(r_{(s, t)})^{1-n}} = 1 = \int_{t\in\mathbb{R}} t^n d\delta_1=\int_{t\in\mathbb{R}} t^nd\mu_{\vec{G}_{triv}}, 
    \end{equation*}
    which shows the truth of the base case in the induction. Now assume that $\varphi_n(\mu_{\vec{G}_{A-B}})=\Tilde{m}_{n;\vec{G}_{A-B}}$ is true for all quotient graphs with the number of edges less than $k$. Then, consider a series-parallel quotient graph $\vec{G}_{A-B}$ with $k$ edges s.t. $\vec{G}_{A-B}=\vec{H}_{1;A-B}\parallelcon\vec{H}_{2;A-B}$. As $\boldsymbol{\beta}\in\operatorname{Ord}(V(\vec{G}_{A-B}), \mathcal{S}_n)$ can be split into two independent order morphisms $\boldsymbol{\beta}_1\in\operatorname{Ord}(V(\vec{H}_{1;A-B}), \mathcal{S}_n)$ and $\boldsymbol{\beta}_2\in\operatorname{Ord}(V(\vec{H}_{2;A-B}), \mathcal{S}_n)$ and $r_{\vec{G}_{A-B}}=r_{\vec{H}_{1;A-B}}\cdot r_{\vec{H}_{2;A-B}}$ by definition, using \cref{eqn:moments of limit distribution}, we can find that 
    \begin{multline*}
        \Tilde{m}_{n;\vec{G}_{A-B}} = \Tilde{m}_{n;\vec{H}_{1;A-B}}\cdot\Tilde{m}_{n;\vec{H}_{2;A-B}} \\
        = \int_{t\in\mathbb{R}} t^nd\mu_{\vec{H}_{1;A-B}}\cdot\int_{t\in\mathbb{R}} t^nd\mu_{\vec{H}_{2;A-B}}\\
        =\int_{t\in\mathbb{R}} t^n d\mu_{\vec{H}_{1;A-B}}\ast\mu_{\vec{H}_{2;A-B}}=\int_{t\in\mathbb{R}} t^n d\mu_{\vec{G}_{A-B}}.
    \end{multline*}
    This proves the case for the parallel composition. Finally, consider a series-parallel quotient graph $\vec{G}_{A-B}$ with $k$ edges such that $\vec{G}_{A-B}=\vec{H}_{1;A-C}\seriescon\vec{H}_{2;C-B}$. In this case, $\boldsymbol{\beta}\in\operatorname{Ord}(V(\vec{G}_{A-B}), \mathcal{S}_n)$ can be also split into two order morphisms $\boldsymbol{\beta_1}:(V(\vec{H}_{1;A-C}), \leq_{\mincut})\to(\mathcal{S}_n, \leq_{\mathrm{nc}})$ and $\boldsymbol{\beta_2}:(V(\vec{H}_{2;C-B}), \leq_{\mincut})\to(\mathcal{S}_n, \leq_{\mathrm{nc}})$ under that condition that 
    \begin{equation*}
        \boldsymbol{\beta}_{1}(A)=\mathrm{id}\leq_{\mathrm{nc}}\boldsymbol{\beta}_1(C)=\boldsymbol{\beta}_2(C)\leq_{\mathrm{nc}}\boldsymbol{\beta}_2(B)=\gamma
    \end{equation*}
    We write $\boldsymbol{\beta}_1(C)=\boldsymbol{\beta}_2(C)=\beta_C$, and then \cref{eqn:moments of limit distribution} implies that 
    \begin{equation*}
        \begin{split}
            \Tilde{m}_{n;\vec{G}_{A-B}} &= \sum_{\mathrm{id}\leq_{\mathrm{nc}}\beta_C\leq_{\mathrm{nc}}\gamma}\left(\frac{r_{\vec{H}_{1;A-C}}}{r_{\vec{G}_{A-B}}}\right)^{\#\beta_C-n}\Tilde{m}_{\beta_C;\vec{H}_{1;A-C}}\cdot \left(\frac{r_{\vec{H}_{2;C-B}}}{r_{\vec{G}_{A-B}}}\right)^{\#(\beta_C^{-1}\gamma)-n}\Tilde{m}_{\beta_C^{-1}\gamma;\vec{H}_{2;C-B}} \\
            &= \sum_{\mathrm{id}\leq_{\mathrm{nc}}\beta_C\leq_{\mathrm{nc}}\gamma}\left(\frac{r_{\vec{H}_{1;A-C}}}{r_{\vec{G}_{A-B}}}\right)^{\#\beta_C-n}\varphi_{\beta_C}(\mu_{\vec{H}_{1;A-C}})\left(\frac{r_{\vec{H}_{2;C-B}}}{r_{\vec{G}_{A-B}}}\right)^{\#(\beta_C^{-1}\gamma)-n}\varphi_{\beta_C^{-1}\gamma}(\mu_{\vec{H}_{2;C-B}})
        \end{split}
    \end{equation*}
    where we have used the moment $\varphi_{\alpha}(\mu)$ from the extended family --- for each permutation $\alpha$, we associate a moment which is a product of the moments each with the order of the cardinality of a cycle of $\alpha$. The relation between the extended family of moments and the extended family of free cumulants is known as (see \cref{appendix:free-probability})
    \begin{equation*}
        \varphi_{\beta}(\mu) = \sum_{\alpha\leq_{\mathrm{nc}}\beta}\kappa_{\alpha}(\mu), 
    \end{equation*}
    and it further implies the following relation that we need to use in our proof. 
    \begin{equation*}
        \begin{split}
            r^{\#\beta-n}\varphi_{\beta}(\mu) &= \sum_{\alpha\leq_{\mathrm{nc}}\beta}r^{\#\beta-\#\alpha}r^{\#\alpha-n}\kappa_{\alpha}(\mu) \\ 
            &=\sum_{\alpha\leq_{\mathrm{nc}}\beta}r^{\#(\alpha^{-1}\beta)-n}\kappa_{\alpha}((r^{-1}\odot\mu)^{\boxplus r}) \\ 
            &=\sum_{\alpha\leq_{\mathrm{nc}}\beta}\kappa_{\alpha^{-1}\beta}(\hat{\pi}_r)\kappa_{\alpha}((r^{-1}\odot\mu)^{\boxplus r}) \\ 
            &=\kappa_{\beta}(\hat{\pi}_r\boxtimes(r^{-1}\odot\mu)^{\boxplus r}).
        \end{split}
    \end{equation*}
    Let $r_1=\frac{r_{\vec{H}_{1;A-C}}}{r_{\vec{G}_{A-B}}}$, $\mu_1=\mu_{\vec{H}_{1;A-C}}$, $r_2=\frac{r_{\vec{H}_{2;C-B}}}{r_{\vec{G}_{A-B}}}$, and $\mu_2=\mu_{\vec{H}_{2;C-B}}$ and note that either $r_1=1$ or $r_2=1$ since $r_{\vec{G}_{A-B}}=\min\{r_{\vec{H}_{1;A-C}}, r_{\vec{H}_{2;C-B}}\}$ by definition. We can then find that 
    \begin{equation*}
        \begin{split}
            \Tilde{m}_{n;\vec{G}_{A-B}} &= \sum_{\mathrm{id}\leq_{\mathrm{nc}}\beta_C\leq_{\mathrm{nc}}\gamma}\kappa_{\beta_C}(\hat{\pi}_{r_1}\boxtimes(r_1^{-1}\odot\mu_1)^{\boxplus r_1})\kappa_{\beta_C^{-1}\gamma}(\hat{\pi}_{r_2}\boxtimes(r_2^{-1}\odot\mu_2)^{\boxplus r_2}) \\ 
            &= \kappa_n(\hat{\pi}_{r_1}\boxtimes(r_1^{-1}\odot\mu_1)^{\boxplus r_1}\boxtimes\hat{\pi}_{r_2}\boxtimes(r_2^{-1}\odot\mu_2)^{\boxplus r_2}) \\ 
            &= \varphi_n((r_1^{-1}\odot\mu_1)^{\boxplus r_1}\boxtimes\hat{\pi}_{\max\{r_1, r_2\}}\boxtimes(r_2^{-1}\odot\mu_2)^{\boxplus r_2}) \\ 
            &= \varphi_n(\mu_{\vec{G}_{A-B}}).
        \end{split}
    \end{equation*}
    Hence, by induction, we have shown that $\varphi_n(\mu_{\vec{G}_{A-B}})=\Tilde{m}_{n;\vec{G}_{A-B}}$ is true for any series-parallel quotient graph $\vec{G}_{A-B}$. The theorem follows from the fact that 
    \begin{equation*}
        \varphi_n(\mu_{\vec{G}_{A-B}})=\Tilde{m}_{n;\vec{G}_{A-B}} = \Tilde{m}_{n;G_{A-B}} = (r_{\vec{G}_{A-B}})^{n-1}m_{n}.
    \end{equation*}
\end{proof}

Combined with \cref{theorem:order morphism}, one can see that the $n$-th moment of the measure to a series-parallel graph $G$ from \cref{def:measure} can be interpreted as a weighted sum of order morphisms from $(V(G), \leq_{\mincut})$ to $(\mathcal{S}_n, \leq_{\mathrm{nc}})$ and in particular, when all of the edges in $G$ are associated with vector spaces of the same dimension, the $n$-th moment of the measure simply counts the number of order morphisms, as indicated by \cref{corollary:order morphisms}. \cref{theorem:relation to random tensor networks} thus leads to a combinatorial interpretation to the moments of the graph-dependent measure. On the other hand, if we demand \cref{theorem:relation to random tensor networks} holds for any graph, which are not necessarily series-parallel, the recursive relation in \cref{def:measure} still holds, though the measure cannot be decomposed into a set of measures on trivial graphs. A significant open question is how to formally define the graph-dependent measures for arbitrary graphs. Unlike the series-parallel case, general graphs lack a simple recursive decomposition, which prevents the direct application of powerful tools such as free convolution. As a result, new techniques or frameworks may be necessary to extend these constructions to arbitrary graph topologies, making this an intriguing and challenging direction for future research.

There is also a close connection to \cite{cheng_random_2022}. In the case where a random tensor network has a quotient graph that is a series composition of two trivial graphs. We can obtain the distribution of the random tensor represented by this tensor network using \cite[Theorem 3.4]{cheng_random_2022} as $\pi_1\boxtimes\delta_1\boxtimes((1-r_G^{-1})\odot\delta_0+r_G^{-1}\odot\delta_1)$. This distribution is consistent with the graph-dependent measure for the quotient graph since $\hat{\pi}_{r_G} = \pi_1\boxtimes\delta_1\boxtimes((1-r_G^{-1})\odot\delta_0+r_G^{-1}\odot\delta_1)$. 

\bigskip

\noindent\textbf{Acknowledgments.} The authors would like to thank Razvan Gurau for insightful discussions about the topics of the paper. Both authors were supported by the ANR project \href{https://esquisses.math.cnrs.fr/}{ESQuisses}, grant number ANR-20-CE47-0014-01.

\bigskip

\printbibliography

\appendix 
\section{Free probability}
\label{appendix:free-probability}
We review the basic definitions and theorems in free probability theory here. We consider in particular random variables coming from a $C^{\ast}$-probability space $(\mathcal{A}, \varphi)$, where $\mathcal{A}$ is a $C^{\ast}$-algebra over $\mathbb{C}$ and $\varphi$ is a linear functional $\varphi:\mathcal{A}\to\mathbb{C}$ which is 
\begin{itemize}
    \item unital: $\varphi(1_{\mathcal{A}})=1$ where $1_{\mathcal{A}}\cdot a=a\quad\forall a\in\mathcal{A}$;
    \item tracial:. $\varphi(ab)=\varphi(ba)\quad\forall a, b\in\mathcal{A}$;
    \item positive: $\varphi(a^{\ast}a)\geq0\quad\forall a\in\mathcal{A}$. 
\end{itemize}
The functional $\varphi$ plays the role of the expectation of a random variable from the algebra. In the case where $a\in\mathcal{A}$ is self-adjoint, there always exists a measure $\mu_{a}$ supported on $\mathbb{R}$ such that 
\begin{equation*}
    \int_{t\in\mathbb{R}} t^nd\mu_a= \varphi(a^n) \quad\forall n\in\mathbb{N}.
\end{equation*}
Hence we can say that the measure $\mu_a$ is the distribution of the random variable $a$ as all of the moments $\varphi(a^n)$ can be extracted from the measure $\mu_a$. 

Similar to classical probability theory, a set of self-adjoint random variables $\{a_i\}_{i\in I}$ is called \emph{free} or \emph{freely independent} if for any $\{i_j\in I\}_{j=1, \dots, n}$ such that $i_1\neq i_2\neq \cdots \neq i_n$, we have 
\begin{multline*}
    \varphi(p_1(a_{i_1})\cdots p_n(a_{i_n}))=0 \text{ for all polynomials $p_1, \dots, p_n$}\\
    \text{whenever }\varphi(p_j(a_{i_j}))=0 \quad\forall j\in\{1, \dots, n\}.
\end{multline*}
As an example, given two self-adjoint random variables $\{a_1, a_2\}$, if they are free, then  
\begin{equation*}
    0=\varphi((a_1-\varphi(a_1))(a_2-\varphi(a_2)))=\varphi(a_1a_2)-\varphi(a_1)\varphi(a_2), 
\end{equation*}
which indeed generalizes the classical independence. In free probability, one often uses an alternative formulation for freely independence, involving the so-called \emph{free cumulants} \cite[Lecture 11]{nica_lectures_2006}. To introduce the free cumulants, we first start by reviewing the non-crossing partitions of an ordered set. Recall that, a partition of a set consists of blocks of disjoint subsets which covers the whole set by taking the union. For an ordered set of size $n$, a partition of the set then naturally corresponds to a permutation $\alpha\in\mathcal{S}_n$ where a block of the partition corresponds to a cycle of $\alpha$. A partition is \emph{non-crossing} if there do not exist $i<j<k<l$ in the ordered set such that $i$ and $k$ belong to one block of the partition while $j$ and $l$ belong to another. Clearly, $\gamma=(1\cdots n)\in\mathcal{S}_n$ corresponds to a non-crossing partition of the set $(1, \dots, n)$ and if $\alpha\in\mathcal{S}_n$ corresponds to another non-crossing partition, then $\alpha$ must be on the geodesic from $\mathrm{id}$ to $\gamma$ \cite[Lecture 23]{nica_lectures_2006}. In fact, the permutations which correspond to non-crossing partitions are partially ordered. We write $\alpha_1\leq_{\mathrm{nc}}\alpha_2$ if $|\alpha_1|+|\alpha_1^{-1}\alpha_2|=|\alpha_2|$ (recall that $|\cdot|$ is Cayley distance) and it is easy to see that $\leq_{\mathrm{nc}}$ is reflexive, antisymmetric, and transitive. Then, the free cumulant $\kappa_{\alpha}[a_1, \dots, a_n]$ is a polynomial of $a_1, \dots, a_n$ indexed by the permutation $\alpha$ such that   
\begin{equation}
    \label{eqn:moment-cumulant formula}
    \varphi(a_1\cdots a_n) = \sum_{\alpha\leq_{nc}\gamma}\kappa_{\alpha}[a_1, \dots, a_n].
\end{equation}
One can invert the above equation using M\"obius inversion for obtaining the definition of the free cumulants, which we omit here. The point of introducing the free cumulants is that if a set of self-adjoint random variables $\{a_i\}_{i\in I}$ is free, then all of the mixed free cumulants vanish, i.e. 
\begin{equation*}
    \forall n\geq 2:\quad \kappa_{\gamma}[a_{i_1}, \dots, a_{i_n}] = 0\text{ whenever $\exists l, k\in\{1, \dots, n\}: i_l\neq i_k$}. 
\end{equation*}
For simplifying notations, when all of the arguments in $\kappa_{\gamma}[\cdot]$ are the same, we use the following 
\begin{equation*}
    \kappa_n(\mu_a)\coloneqq \kappa_{\gamma}[a, \dots, a]|_{\gamma=(1\cdots n)} \quad\text{and}\quad \kappa_{\alpha}(\mu_a)\coloneqq\prod_{\sigma:\text{cycle of $\alpha$}}\kappa_{\operatorname{card}(\sigma)}(\mu_a). 
\end{equation*}
It is not hard to see from \cref{eqn:moment-cumulant formula} that these cumulants are related to the extended moments $\varphi_{\beta}(\mu_a)$ ($\beta\in\mathcal{S}_n$) of the distribution $\mu_a$ as follows: 
\begin{equation*}
    \varphi_{\beta}(\mu_a)\coloneqq \prod_{\sigma:\text{cycle of $\beta$}}\int_{t\in\mathbb{R}} t^{\operatorname{card}(\sigma)}d\mu_a=\sum_{\alpha\leq_{\mathrm{nc}}\beta}\kappa_{\alpha}(\mu_a). 
\end{equation*}

The free additive convolution $\boxplus$ is a binary operation introduced to combine the distributions of two freely independent random variables additively. More precisely, let $a$ and $b$ be two freely independent self-adjoint random variables with distributions $\mu_a$ and $\mu_b$ supported on $\mathbb{R}$ respectively. Then, the free additive convolution $\mu_a\boxplus\mu_b$ is defined as the distribution of the random variable $a+b$, which is also supported on $\mathbb{R}$. The free cumulants of these measure satisfies the following relation 
\begin{equation*}
    \kappa_{n}(\mu_a\boxplus\mu_b) = \kappa_n(\mu_a) + \kappa_n(\mu_b).
\end{equation*}
For analytical computation, one can use the following theorem which relates these cumulants with the Cauchy transform of their corresponding measures and hence obtain the measure $\mu_a\boxplus\mu_b$ with the help of Stieltjes inversion. 
\begin{theorem}{\cite[Lecture 16]{nica_lectures_2006}}
    The $R$-transform of $\mu$ is the formal power series 
    \begin{equation*}
        R_{\mu}(z) \coloneqq \sum_{n\geq 1}\kappa_{n}(\mu)z^n,
    \end{equation*}
    which can be also obtained by solving the following equation 
    \begin{equation*}
        G_{\mu}[(R_{\mu}(z)+1)/z]=z\text{, with $G_{\mu}(z) \coloneqq \int_{t\in\mathbb{R}}\frac{d\mu}{z-t}$ the Cauchy transformation of $\mu$}.
    \end{equation*}
    The $R$-transform is additive under the free additive convolution, i.e. 
    \begin{equation*}
        R_{\mu_a\boxplus\mu_b}(z) = R_{\mu_a}(z) + R_{\mu_b}(z).
    \end{equation*}
\end{theorem}
As an example, the $R$-transform of the free Poisson distribution with rate $r$ and jump size $s$, 
\begin{equation*}
    \pi_{r; s} \coloneqq \lim_{N\to\infty}\left(\left(1-\frac{r}{N}\right)\delta_0+\frac{r}{N}\delta_{s}\right)^{\boxplus N}, 
\end{equation*}
is given as 
\begin{equation*}
    R_{\pi_{r;s}}(z) = \sum_{n\geq 1}\kappa_n(\pi_{r;s})z^n = \sum_{n\geq 1}rs^nz^n = \frac{rs z}{1-s z}. 
\end{equation*}
In particular, the free Poisson distribution with rate $r$ and jump size $s=1$ is also Marchenko-Pastur distribution we have introduced in \cref{sec:graph-dependent measures}, where we write $\pi_{r;1}$ as $\pi_r$ and $r^{-1}\odot\pi_r$ as $\hat{\pi}_{r}$. 

Similarly, in the multiplicative case, we consider $a$ and $b$ to be two freely independent \emph{positive} self-adjoint random variables with distributions $\mu_a$ and $\mu_b$ supported on $\mathbb{R}^{+}$. The free multiplicative convolution $\mu_a\boxtimes\mu_b$ is defined as the distribution of the random variable $\sqrt{a}b\sqrt{a}$ (or equivalently $\sqrt{b}a\sqrt{b}$), which is also supported on $\mathbb{R}^{+}$ as $\sqrt{a}b\sqrt{a}$ is positive and self-adjoint. Note that the moments of $\sqrt{a}b\sqrt{a}$ are the same as those of $ab$ while, unlike the former, the latter is not, in general, self-adjoint given $a$ and $b$ self-adjoint. The free cumulants of these measure satisfy the following relation
\begin{equation*}
    \kappa_n(\mu_a\boxtimes\mu_b) = \sum_{\alpha\leq_{\mathrm{nc}}\gamma}\kappa_{\alpha}(\mu_a)\kappa_{\alpha^{-1}\gamma}(\mu_b),
\end{equation*}
where $\alpha^{-1}\gamma$ is known as the Kreweras complement of $\alpha$. For analytical computation, one often needs to compute the $S$-transforms of two measures to obtain their convolution. 
\begin{theorem}{\cite[Lecture 18]{nica_lectures_2006}}
    The $S$-transform of $\mu$ is 
    \begin{equation*}
        S_{\mu}(z) \coloneqq \frac{1}{z}R_{\mu}^{-1}(z),
    \end{equation*}
    where $R_{\mu}^{-1}$ is the inversion of $R_{\mu}$ under function composition. The $S$-transform is multiplicative under the free multiplicative convolution, i.e. 
    \begin{equation*}
        S_{\mu_a\boxtimes\mu_b}(z)=S_{\mu_a}(z)S_{\mu_b}(z).
    \end{equation*}
\end{theorem}
As an example, the $S$-transform of the free Poisson distribution is given as 
\begin{equation*}
    S_{\pi_{r;s}}(z)=\frac{1}{s(r+z)}. 
\end{equation*}
and the $S$-transform of the rescaled Marchenko-Pastur distribution $\hat{\pi}_r$ is given as 
\begin{equation*}
    S_{\hat{\pi}_r}(z) = \frac{r}{r+z}. 
\end{equation*}
Finally, the free additive and multiplicative convolutions allow us to relate the cumulants on a projected subspace with those on the whole space, which would explain \cref{def:measure}. 
\begin{theorem}{\cite[Theorem 14.10]{nica_lectures_2006}}
    \label{theorem:projection}
    Consider a $C^{\ast}$-probability space $(\mathcal{A}, \varphi)$ and random variables $p, a\in\mathcal{A}$ such that $p$ is a projection and freely independent from $a$. Then, on the $C^{\ast}$-probability subspace $(p\mathcal{A}p, \varphi^{(p\mathcal{A}p)}\coloneqq r\varphi)$, one can find that  
    \begin{equation*}
        \kappa_n^{(p\mathcal{A}p)}(\mu_{pap}) = \kappa_n((r^{-1}\odot\mu_a)^{\boxplus r}), 
    \end{equation*}
    where $r=1/\varphi(p)$ so that $\varphi^{(p\mathcal{A}p)}(p1_{\mathcal{A}}p)=1$ as $\varphi^{(p\mathcal{A}p)}$ is required to be unital. 
\end{theorem}

\end{document}